\documentclass[11pt,reqno,oneside]{amsart}

\usepackage{graphicx}
\usepackage{amssymb}
\usepackage{epstopdf}
\usepackage{cite}
\usepackage[a4paper, total={6in, 9in}]{geometry}
\usepackage{amsmath,amsfonts,amsthm,mathrsfs,amssymb,cite,dsfont}
\usepackage[usenames]{color}
\usepackage{bm}
\usepackage{subfigure}

\usepackage{algorithm}
\usepackage{algpseudocode}
\algrenewcommand\algorithmicrequire{\textbf{Input:}}
\algrenewcommand\algorithmicensure{\textbf{Output:}}

\usepackage{multirow}
\makeatletter
\def\BState{\State\hskip-\ALG@thistlm}
\makeatother

\usepackage{amsmath,amsfonts,amsthm,mathrsfs,amssymb,cite,mathrsfs}
\usepackage[usenames]{color}

\newtheorem{thm}{Theorem}[section]

\theoremstyle{definition}
\newtheorem{defn}{Definition}[section]

\theoremstyle{remark}

\newtheorem{rem}{Remark}[section]
\numberwithin{equation}{section}
\numberwithin{equation}{section}

\newtheorem{exam}{Example}[section]

\newcounter{saveeqn}

\title[A novel approach for geometrical body generation]{An inverse scattering approach for geometric body generation: a machine learning perspective }

\author{Jinghong Li}
\address{Faculty of Science, Qilu University of Technology, Jinan, Shandong, China}
\email{lijinhong@qlu.edu.cn}

\author{Hongyu Liu}
\address{Department of Mathematics, Hong Kong Baptist University, Kowloon, Hong Kong SAR}
\email{hongyu.liuip@gmail.com}

\author{Wing-Yan Tsui}
\address{Department of Mathematics, Hong Kong Baptist University, Kowloon, Hong Kong SAR}
\email{wytsui.yan@gmail.com}

\author{Xianchao Wang}
\address{Department of Mathematics, Harbin Institute of Technology, Harbin}
\email{xcwang90@gmail.com}

\date{} 
\begin{document}
\maketitle

\begin{abstract}

In this paper, we are concerned with the 2D and 3D geometric shape generation by prescribing a set of characteristic values of a specific geometric body. One of the major motivations of our study is the 3D human body generation in various applications. We develop a novel method that can generate the desired body with customized characteristic values. The proposed method follows a machine-learning flavour that generates the inferred geometric body with the input characteristic parameters from a training dataset. The training dataset consists of some preprocessed body shapes associated with appropriately sampled characteristic parameters. One of the critical ingredients and novelties of our method is the borrowing of inverse scattering techniques in the theory of wave propagation to the body generation. This is done by establishing a delicate one-to-one correspondence between a geometric body and the far-field pattern of a source scattering problem governed by the Helmholtz system. It in turn enables us to establish a one-to-one correspondence between the geometric body space and the function space defined by the far-field patterns. Hence, the far-field patterns can act as the shape generators. The shape generation with prescribed characteristic parameters is achieved by first manipulating the shape generators and then reconstructing the corresponding geometric body from the obtained shape generator by a stable multiple-frequency Fourier method. The proposed method is in sharp difference from the existing methodologies in the literature, which usually treat the human body as a suitable Riemannian manifold and the generation is based on non-Euclidean approximation and interpolation. Our method is easy to implement and produces more efficient and stable body generations. We provide both theoretical analysis and extensive numerical experiments for the proposed method. The study is the first attempt to introduce inverse scattering approaches in combination with machine learning to the geometric body generation and it opens up many opportunities for further developments.

\medskip



\end{abstract}

\section{Introduction}


With the rapid technological advancement today, the access to realistic 3D human shapes is of great importance in both computer vision and graphics, and has various applications in different industries including virtual game design, film making, bioinformatics\cite{BioDigital}, healthcare\cite{treleaven20073d}, and especially, those related to garment design. Some applications involve fitting predictions, virtual try-on simulations\cite{guo2012clothed,chen2011practical,zhou2010parametric,tong2012scanning} or size recommendations\cite{ashdown2004using,apeagyei2010application,loker2005size}, that help to recommend relevant clothing which would fit specific occasions or fashion trends for online customers. Such applications require a critical ingredient on digital transformation from humans bodies to digital 3D shapes, such that the shapes maintain some of the main features from human bodies.

The traditional approaches to access reliable digital information of a human body are through laser range scanners\cite{levoy2000digital}, stereo reconstruction\cite{seitz2006comparison,esteban2004silhouette,yang2007reconstruction} or structured light methods for 3D sensing\cite{geng2011structured,chen2000overview,dunn1989measuring}. However, considering the cost of data storage, network transmission and expensive scanning equipment, it is rather unpractical to scan individuals for each application. Hence many studies have been done to generate 3D human shapes based on partial input information. These prior systems can be mainly classified into three types: marker-based systems, silhouette-based systems  and  measurement-based systems. Marker-based system estimates dynamic 3D human body shapes by capturing a sparse set of marker positions. These techneqiues proceed by using a single static scan or multiple scans and a marker motion capture sequence of the person\cite{anguelov2005scape}. For the static case, silhouette-based system estimates human body shapes based on a set of input images by fitting the silhouette in each view\cite{balan2007detailed,chen2011single,guan2009estimating,mundermann2007accurately}. Some apporaches also combine with machine learning that build a correlation between a training dataset of 3D body shapes and a set of 2D images, and then predict a shape based on the correlation\cite{chen2009learning}.

Although marker-based systems and silhouette-based systems could yield satisfactory reconstructions on 3D human body shapes under tight dresses or naked human shapes, most of the schemes are so computationally expensive and the results are easily affected if heavy or loose clothes are worn. To overcome these difficulties, a great deal of efforts have been devoted to the investigation of simple and fast measurement-based systems \cite{seo2003automatic,seo2003synthesizing,kasap2007parameterized}. Typically, one considers the landmarks or circumferences from the human structures at specific locations as characteristic values. Since such characteristic values are linear or curvilinear, they are relatively invariant to articulation changes than those silhouettes measurements. If the set of characteristic values is well selected, one can achieve meaningful estimation for both global and local body shapes. Kart et al.(2011) built a system which only requires to input some personal information, such as weight, height and age as well as a 2D photograph. The decision algorithm then determines the human shape according to the measurements and the body mass index (BMI)\cite{kart2012web}.  Seo et al. \cite{seo2003automatic} presented a human body generation apporach by taking the anthropometric measurements, e.g. stature, crotch length, arm length, neck girth, chest/bust girth, underbust girth, waist girth and hip girth as input. They derived the relationship between the input characteristic values and the preprocessing database of 3D scanned data of human body models by using radial basis interpolation. At run-time, the system generates new human body shapes from the user input characteristic values by fitting the template model onto each scanned data.

In this paper, we develop a completely novel methodology for the geometric body generation, which fulfils the following two basic requirements: (i) the geometric body generation is automatically determined by the input characteristic sets; (ii) the predicted geometric shape fits for all input characteristic values and moreover it can well approximate the exact geometric body possessing the aforesaid characteristic values. The proposed method follows a machine-learning flavour that generates the inferred geometric body with the customized characteristic parameters from a training dataset. The training dataset consists of some preprocessed body shapes associated with appropriately sampled characteristic parameters. One of the critical ingredients and novelties of our method is the borrowing of inverse scattering techniques in the theory of wave propagation to the body generation. This is done by establishing a delicate one-to-one correspondence between a geometric body and the far-field pattern of a source scattering problem governed by the Helmholtz system. It in turn enables us to establish a one-to-one correspondence between the geometric body space and the function space defined by the far-field patterns. Hence, the far-field patterns can act as the shape generators. The shape generation with prescribed characteristic parameters is achieved by first manipulating the shape generators in the function space and then reconstructing the corresponding geometric body from the obtained shape generator by a stable multiple-frequency Fourier method. The proposed method is in sharp difference from the existing methodologies in the literature, which usually treat the human body as a suitable Riemannian manifold and the generation is based on non-Euclidean approximation and interpolation. In fact, in all of the literature mentioned earlier on manifold learning of body generation, one typically uses Principal Component Analysis (PCA) or Principal Geodesic Analysis (PGA). {PCA and PGA are used for optimal reduction of the data and thus efficient deformation by computing statistics on Euclidean manifolds or non-Euclidean manifolds can be achieved; see \cite{seo2003automatic,freifeld2012lie} and the references therein for more relevant discussion. In our new approach, the shape generator enables us to train the learning dataset via the algebraic operations in the shape space directly without dealing with the deformation of the manifold meshes between geometric shapes.}

The rest of the paper is organized as follows. In section 2, we provide rigorous mathematical formulations of characteristic values and shape space. Section 3 introduces the notion of shape generator via the inverse source scattering associated with the Helmholtz system. In Section 4, we present the mathematical setup of the geometric body generation from a machine-learning perspective. Section 5 is devoted to the development of the new method for the shape generation. In Section 6, we present several two- and three-dimensional numerical examples to show the effectiveness and efficiency of our method. The paper is concluded in Section 7 with some relevant discussion.

\section{Preliminary knowledge on shape manifold theory}

In this section, we present some preliminary knowledge on the shape manifold theory that shall be needed in our subsequent study of body generation. Generally speaking, a geometric shape or a geometric body is a topological $n$-manifold, $n\in\mathbb{N}$, equipped with certain shape descriptors, which give the full information to describe the geometric shape. We call such shape descriptors as {\it characteristic values}. We have the following formal definition.

\begin{defn}\label{char values}
Let $D$ be a topological $n$-manifold with $n\in\mathbb{N}$. Let $\Lambda_D:=\{\lambda^{(j)}\}_{j\in\mathscr{C}}$ be a set of parameters associated with $D$ that are invariant with respect to isometric deformations and are independent to the parametrizations of $D$. Here, the cardinality $\mathscr{C}$ might be finite or infinite. $\Lambda_D$ is said to be a characteristic set of $D$ if it uniquely determines $D$. $D$ and its characteristic set $\Lambda_D$, written as the  $(D, \Lambda_D)$ is referred to as a geometric shape or a geometric body.
\end{defn}

Clearly, Definition~\ref{char values} includes much general geometric objects. However, for the present study, we are mainly concerned with the case that $D$ can be embedded into $\mathbb{R}^d$, $d=2,3$, as a bounded domain. That means, we exclude some interesting cases such as $D$ is a Riemannian surface with boundary in $\mathbb{R}^3$. Nevertheless, our study is general enough to include the human body as a specific case.

In Definition~\ref{char values}, the set of characteristic values is typically a set of measurements which gives a systematic characterization of the size, shape and composition of a geometric object for us to determine the shape of the object. For example, when considering a rectangular object, once can introduce a set of characteristic values containing its height, width and length, which provide all details to determine a unique rectangular shape. Expanding the same idea to human body shapes, one could also use characteristic sets to represent them.  There are many different ways to represent a human shape. We would try to group those characteristic values into four main catagories, including Eucidean distance, geodesic distance, circumference and ratio. The Eucidean or geodisic distance is linear or curvilinear distance between two points on the human model, such as stature, crotch length, arm length, shoulder breadth etc.. The circumference can be computed by the horizontal girth of the body, such as neck girth, chest/bust girth, under-bust girth, waist girth, hip girth, etc.. The ratio can be information of weight, Body Mass Index, muscle and fat rate. In spite of the above characteristic values, one can also consider some pure measurements such as age or gender as characteristic values.

The full set of characteristic values gives the complete information of a geometric shape without lossing any information. It is easy to imagine that the cardinality of a set of characteristic values depends on the complexity of a shape. Hence, the number of characteristic values required can be considered as the dimensionality of the geometric shape. For those complicated objects, like human shapes, it may require infinite set of characteristic values for accurate formulations. Due to practical reasons, one can consider a reasonable truncation of an infinite characteristic set into a finite one for a complicated geometric shape. In doing so, we can consider our study in the following product space
\begin{equation}\label{eq:product1}
\mathcal{S}:=\mathscr{D}\times \mathscr{V},
\end{equation}
where $\mathscr{D}$ is composed of all the bounded domains in $\mathbb{R}^d$ and $\mathscr{V}$ is an $M$-dimensional vector space containing the characteristic values. In fact, in the present study, the characteristic values are usually real numbers and one can take $\mathscr{V}=\mathbb{R}^{M}$ with $M\in\mathbb{N}$. $\mathcal{S}$ is referred to as the geometric shape space. According to the (approximate) one-to-one correspondence between a geometric shape and its characteristic values in Definition~\ref{char values}, we readily see that all the shape information can be obtained by a single point of this $M$-dimensional vector space $\mathscr{V}$. By adjusting the characteristic values, we can obtain new geometric shapes and this is a key ingredient in our human body interpolation.

\section{Shape generators via inverse source scattering}\label{shape generator}

In the previous section, we introduce the important notion of shape space for our study. We proceed to introduce another critical ingredient, {\it shape generator}, for our subsequent study of the geometric body generation. In fact, the generation of a new geometric shape shall be based on algebraic interpolation of exemplar models from the shape space. If the algebraic operations are to be conducted directly in the shape space, dealing with geometric deformations of manifolds, one would certainly encounter very complicated and tedious calculations and manipulations because of the lack of global parametrizations for the non-Eucidean shapes involved. The shape generator can overcome this challenge by bridging the geometric shape space and the function space. To that end, we next introduce the inverse scattering problem in finding an active source from its generated far-field pattern.

Let $f:\mathbb{R}^d\mapsto\mathbb{C}$ be a function having a compact support, $f =
\chi_D \varphi$, where $D\subset\mathbb{R}^d$ is a bounded
domain and $\varphi\in
L^\infty(\mathbb{R}^d)$. The set $D$ is the external \emph{shape} of $f$
while $\varphi$ describes the \emph{intensity} of the source at
various points in $D$. We assume that $\varphi$ and $D$ do
not depend on the wavenumber $k\in\mathbb{R}_+$. In other words we are
considering monochomatic scattering. The source $f$ produces a
scattered wave $u\in H^2_{loc}(\mathbb{R}^d)$ given by the unique solution to
\begin{equation} \label{sourceScattering}
  (\Delta+k^2) u = f, \qquad \lim_{r\rightarrow\infty} r^\frac{d-1}{2}
  \big(\partial_r - ik \big) u = 0,
\end{equation}
where $r=|x|$ for $x\in\mathbb{R}^d$. The limit in \eqref{sourceScattering} is known as the Sommerfeld radiation condition which characterizes the outgoing nature of the radiating wave. By the limiting absorption principle (cf. \cite{Eskin}), the solution to \eqref{sourceScattering} can be computed as follows,
\begin{equation}\label{eq:source1}
\begin{split}
u=(\Delta+k^2)^{-1} f & =\lim_{\varepsilon\rightarrow +0} \big(\Delta+(k-i\varepsilon)^2 \big)^{-1}f\\
&=-\lim_{\varepsilon\rightarrow+0}\int_{\mathbb{R}^d} \frac{e^{i x\cdot\xi}\widehat f(\xi)}{|\xi|^2-(k-i\varepsilon)^2} \ d\xi,
\end{split}
\end{equation}
where
\begin{equation}\label{eq:fourier}
\widehat f(\xi):=\mathcal{F}f(\xi)=(2\pi)^{-d}\int_{\mathbb{R}^n} f(x) e^{-i\xi\cdot x}\ dx
\end{equation}
signifies the Fourier transform of $f$. Inverting the Fourier transform in \eqref{eq:source1}, one has the following integral representation,
\begin{equation}\label{eq:source2}
u=(\Delta+k^2)^{-1} f:=-\frac i 4 \left(\frac{k}{2\pi}\right)^{\frac{d-2}{2}}\int_{\mathbb{R}^d} |x-y|^{\frac{2-d}{2}} H_{\frac{d-2}{2}}^{(1)}(k|x-y|) f(y)\ dy,
\end{equation}
where $H_{(d-2)/2}^{(1)}$ is the first-kind Hankel function of order $(d-2)/2$. Stationary phase applied to \eqref{eq:source2} yields that
\begin{equation}\label{eq:source3}
u(x)=\frac{e^{ik|x|}}{|x|^{(d-1)/2}} C_{d,k}\int_{\mathbb{R}^d} e^{-ik\hat{x}\cdot y} f(y)\ dy+ \mathcal O(|{x}|^{\frac d 2}),\quad |x|\rightarrow \infty,
\end{equation}
where $\hat x:=x/|x|\in\mathbb{S}^{d-1}$, $x\in\mathbb{R}^d\backslash\{0\}$, and
\[
C_{d,k} = \frac{-i}{\sqrt{8\pi}}\left(\frac{k}{2\pi}
\right)^{\frac{d-2}{2}} e^{-\frac{(d-1)\pi}{4}i}.
\]
The \emph{far-field pattern} of $u$ is given by
\begin{equation}\label{eq:sss1}
  u_\infty(\hat x, k; f):= C_{d,k}\int_{\mathbb{R}^d} e^{-ik\hat{x}\cdot y}
  f(y)\ dy=(2\pi)^d C_{d,k} \mathcal{F}f(k\hat x)\in L^2(\mathbb{S}^{d-1}).
\end{equation}
It is obvious that $u_\infty$ is (real) analytic in both $\hat x$ and $k$. Hence, if $u_\infty(\hat x, k)$ is known on any open portion of $\mathbb{S}^{d-1}\times\mathbb{R}_+$, then it is known on the whole set by analytic continuation.

The inverse source scattering problem is concerned with the recovery of $f=\varphi\chi_D$ by knowledge of $u_\infty(\hat x, k; f)$ for $(\hat x, k)\in\Sigma$, where $\Sigma$ is an open subset of $\mathbb{S}^{d-1}\times\mathbb{R}_+$. According to our discussion above, without loss of generality, we always assume that $\Sigma=\mathbb{S}^{d-1}\times\mathbb{R}_+$ in what follows. The inverse source problem arises in a variety of important applications including detection of hazardous chemicals, medical imaging, photoacoustic and thermoacoustic tomography, brain imaging, artificial intelligence in gesture computing and others. We refer to two recent articles \cite{wang2017fourier,wang2018fourier} by two of the authors of this article for some recent developments on the inverse source problem.

Next, let us consider a specific case by assuming a source supported in a domain $D$ with a constant density $1$. Then clearly by \eqref{eq:sss1}, there is a one-to-one correspondence between $D$ and $u_\infty(\Sigma; D):=\{u_\infty(\hat x, k; 1\cdot\chi_D)\}_{(\hat x, k)\in\Sigma}\in L^2(\mathbb{S}^{d-1}\times\mathbb{R}_+)$ in the sense that for two domains $D_1$ and $D_2$,
\begin{equation}\label{eq:11d}
u_\infty(\Sigma; D_1)=u_\infty(\Sigma; D_2)\quad\mbox{if and only if}\quad D_1=D_2.
\end{equation}
Based on \eqref{eq:11d}, we next introduce
\begin{defn}\label{defn:generator}
For a geometric shape $(D, \Lambda_D)\in\mathcal{S}$,
\[
u_\infty(\Sigma; D):=\{u_\infty(\hat x, k; 1\cdot\chi_D)\}_{(\hat x, k)\in\Sigma}\in L^2(\mathbb{S}^{d-1}\times\mathbb{R}_+)
\]
defined via the Helmholtz system \eqref{sourceScattering} is called a {\it shape generator} for $D$.
\end{defn}

\begin{rem}
By Definition~\ref{defn:generator}, a geometric body $D$ can be completely determined by a shape generator $u_\infty(\Sigma)$. Since $u_\infty(\Sigma)$ is from a function space, this paves the way for the new body generation through function interpolations.
\end{rem}

\begin{rem}
By \eqref{eq:sss1}, we know the far-field pattern is actually the Fourier transform of the source density up a dimensional constant. However, introducing the shape generator via the inverse scattering approach shall provide more physical insights in our study, and moreover it enables us to borrow ideas from the inverse scattering literature of recovering the geometric shape $D$ from the associated far-field pattern. This also paves the way of extending the idea by using other inverse scattering models that have such one-to-one correspondence between geometric shapes and far-field patterns; see more relevant discussion in Section 7.
\end{rem}

\section{Mathematical setup for the geometric body generation}

In this section, we introduce the mathematical formulation of the geometric body generation for our study from a machine learning perspective. For a geometric shape $(D, \Lambda_{D})$ with the associated shape generator $u_\infty(\Sigma;D)$, the pair of the high dimensional variables, written as
$\{(\Lambda_{D},u_\infty(\Sigma;D))\}$, is referred to as an {\it input-output pair}. Let $\{(\Lambda_{D_i},u_\infty(\Sigma;D_i))\}_{i\in \mathscr{N}}$ with $\mathscr{N}=\{0,1,\dots,N_{pair}\}$ be a set of input-output pairs associated with the characteristic sets $\{\Lambda_{D_i}\}_{i\in \mathscr{N}}$. Here the input characteristic sets $\Lambda_{D_i}$ are introduced as
\begin{equation}\label{eq:a1}
\Lambda_{D_i}:=\{\lambda^{(j)}_{i}\}_{j\in\mathscr{C},i\in \mathscr{N}},
\end{equation}
with $\mathscr{C}=\{1,\dots,M\}$. In \eqref{eq:a1}, the notation $\lambda^{(j)}_{i}$ represents the characteristic value of the $i$-th geometric shape $D_i$ in the $j$-th direction of its characteristic set. Here, the cardinality $\mathscr{N}$ is finite. The {\it training dataset} of the geometric body generation is introduced to be
	\begin{equation}\label{eq:training dataset}
	\mathbf{Z}:=\{(\Lambda_{D_i},u_\infty(\Sigma;D_i))\}_{i\in \mathscr{N}}
	\end{equation}
	with
	\begin{equation}\label{variable}
	\Lambda=(\lambda^{(1)},\dots,\lambda^{(M)})\in \mathbb{R}^M \text{ and } u_\infty(\Sigma)\in L^2(\mathbb{S}^{d-1}\times\mathbb{R}_+).
	\end{equation}
The training dataset consists of certain pre-sampled geometric shapes with statistically well selected characteristic values. The corresponding shape generator of a specific body in the training dataset can also be pre-calculated and stored. The main goal of our study is to first infer a learning model from the training dataset, $T_\mathbf{Z}: \mathbb{R}^M\rightarrow L^2(\mathbb{S}^{d-1}\times\mathbb{R}_+)$ that fulfils the following requirements:
\begin{enumerate}
\item It fits the training data well in the sense that
\begin{equation}\label{eq:traininng model}
T_\mathbf{Z}(\Lambda_{D_i}):=\hat{u}_\infty(\Sigma;D_i)\approx u_\infty(\Sigma;D_i),\quad \forall D_i\in\mathbf{Z}.
\end{equation}

\item It can be used to infer the shape generator for a given new shape with prescribed characteristic values, namely,
\begin{equation}\label{eq:step1}
\hat{u}_\infty(\Sigma; D_{\mathrm{new}}):=T_\mathbf{Z}(\Lambda_{D_\mathrm{new}}),
\end{equation}
and with a statically well selected training dataset, it is justifiable to expect that
\begin{equation}\label{step2}
\hat{u}_\infty(\Sigma; D_{\mathrm{new}})\approx u_\infty(\Sigma; D_{\mathrm{new}}),
\end{equation}
where $u_\infty(\Sigma; D_{\mathrm{new}})$ is the shape generator for $D_{\mathrm{new}}$.
\end{enumerate}
If a learning model can be achieved that fulfils the two requirements as described above, then the body generation can be proceeded as follows. For a given new set of characteristic values, one first generates the learned shape generator as in \eqref{eq:step1}. By a certain inverse scattering approach, one can then reconstruct the (approximate) shape $D_{\mathrm{new}}$ from the corresponding shape generator $\hat{u}_\infty(\Sigma; D_{\mathrm{new}})$. In the next section, we shall develop the two critical ingredients in the body generation procedure described above, namely, the learning model and the reconstruction method. To be more definite and specific, we first introduce the following definition from a machine learning perspective.
%
%

\begin{defn}\label{Main problem}\title{(\textbf{Body Learning Model})}
Given a {\it training dataset}
\begin{equation}\label{eq:set1}
\mathbf{Z}:=\{(\Lambda_{D_i},u_\infty(\Sigma;D_i))\}_{i\in \mathscr{N}}.
\end{equation}
Let $\mathcal{H}$ be a compact subset of $L^2(\mathbb{S}^{d-1}\times\mathbb{R}_+)$. $T_\mathbf{Z}\in\mathcal{H}$ (with specified coefficients $C$) is said to be the best fit {\it learning model} associated with the training dataset $\mathbf{Z}$ if it is the minimizer of the following optimization problem,
\begin{equation}\label{eq:error function}
\min_{T_\mathbf{Z}\in \mathcal{H}} \frac{1}{N_{pair}+1}\sum_{i\in\mathscr{N}}\left\| T_\mathbf{Z}(\Lambda_{D_i})-u_\infty(\Sigma; D_i)\right\|^2_{\mathcal{H}}.
\end{equation}
\end{defn}


According to Definition~\ref{Main problem}, the choice of the learning subspace $\mathcal{H}$ plays a critical role. However, we note that the shape generator is actually (real) analytic in all of its arguments. Hence, instead of solving the computationally costly optimization problem \eqref{eq:error function}, we can make use of the functional interpolation to produce a well-rounded shape learning model. This is one of the main advantages of introducing the shape generator through the inverse scattering model. In the next section, for a given training dataset as in \eqref{eq:set1}, we shall derive a learning model using the cubic B-spline interpolation through the use of the high-dimensional data-points \eqref{eq:traininng model}. For the reconstruction of the approximate body shape from the shape generator obtained through the learning model, we shall make use of a multiple-frequency Fourier method, and it can also produce an efficient and stable recovery. Throughout, we assume that the characteristic values in the training dataset is statistically well selected and it is not the focus of the present article.

\section{A scheme for geometric body generation}\label{Appaoach}

In this section, we develop the details of our scheme for the geometric body generation following the general discussion made in the previous section. We first derive the learning model through the functional interpolation of the high-dimensional data in the training dataset. To that end, we present some preliminary knowledge on the cubic B-spline, and we also refer to \cite{hou1978cubic,de1972calculating,schoenberg1973cardinal,schoenberg1946contributions} for more relevant discussion on the cubic B-spline.

\subsection{Preliminary knowledge on the cubic B-spline}\label{Non-uniform B-spline}


Consider the training dataset \eqref{eq:training dataset}. Let the $M$ sets of unique {\it grids} in the {directions} of  $\{\lambda^{(j)}\}_{j\in \mathscr{C}}$
\begin{equation}\label{eq:grids}
\Delta_j=\{\lambda_0^{(j)},\cdots,\lambda_{k_j}^{(j)}\}_{j\in \mathscr{C}},\quad k_j\in \mathbb{N},
\end{equation}
define on the intervals $[a_j,b_j]$ as $M$ sets of points $\lambda_{g_j}^{(j)}\in[a_j,b_j]\subset \mathbb{R}$, where $g_j\in \{0,\dots,k_j\}
$ and $a_j=\lambda_{0}^{(j)}< \lambda_{1}^{(j)}< \dots < \lambda_{k_j}^{(j)}=b_j,j=1,\dots,M$. Here, $k_j$ is the greatest number of distinct characteristic values in the $j$-th direction of the characteristic set. We remark that if the characteristic values of the training dataset are all collected in distinct values, then $k_j$ is actually the last index of the training dataset, $N_{pair}$. However, the training dataset might be collected in such a way that some body shapes may possess the same characteristic value in the $j$-th direction, and hence $k_j$ is usually smaller than $N_{pair}$.

With the above notation, the training dataset stored as the array in \eqref{eq:training dataset} can be represented as elements on the grid mesh corresponding to the characteristic numbers $\lambda^{(j)}$ as described above. The interpolation data are the corresponding shape generators and are written as
\begin{equation}\label{eq:interpolation data}
u_\infty(\Sigma;D_i):=U_{g_1,g_2,\dots,g_M}\in L^2(\mathbb{S}^{d-1}\times\mathbb{R}_+),\quad \forall i\in \mathscr{N},
\end{equation}
where $g_j=0,\dots,k_j,j=1,\dots,M.$ In the subsequent study, we shall stick to the same notation $g_1,g_2,\dots,g_M$ to represent the linear indexing. The following example demonstrates a real application for the human body generation.

\begin{exam}\label{exam:example of grids}
	The training dataset consists of 20 bodies with two characteristic values as consideration, say, height and relative weight. Here, the height and relative weight are the two directions of the grids,  $\lambda^{(height)}$ and $\lambda^{(weight)}$. Suppose the heights of the sampled bodies are given by 1.5m, 1.6m, 1.7m, 1.8m, 1.9m and the relative weights of the sampled bodies are given by 60\%, 80\%, 100\%, 120\%. Then the first grid $\Delta_{height}=\{\lambda_0^{(height)},\lambda_1^{(height)},\lambda_2^{(height)},\lambda_3^{(height)},\lambda_4^{(height)}\}=\{1.5,1.6,1.7,1.8,1.9\}$ and the second grid $\Delta_{weight}=\{\lambda_0^{(weight)},\lambda_1^{(weight)},\lambda_2^{(weight)},\lambda_3^{(weight)}\}=\{0.6, 0.8, 1, 1.2\}$. The interpolation data are actually stored as listed in Table~1; e.g. $U_{0,0}=u_\infty(\Sigma;D_1),U_{0,3}=u_\infty(\Sigma;D_4),U_{1,3}=u_\infty(\Sigma;D_{8})$.
\end{exam}

\begin{center}
\begin{table}[]\small
	\hspace*{-.3cm}\begin{tabular}{cc|c|c|c|c|}
		\cline{3-6}
		\multicolumn{1}{l}{}                                                                           & \multicolumn{1}{l|}{}  & \multicolumn{4}{c|}{$2^{nd}$ grid}                                                                                                                                                                                                                                                                                                                                                                                                                  \\ \cline{3-6}
		\multicolumn{1}{l}{}                                                                           &                        & $\lambda_0^{(weight)}$                                                                                      & $\lambda_1^{(weight)}$                                                                                      & $\lambda_2^{(weight)}$                                                                                    & $\lambda_3^{(weight)}$                                                                                      \\ \hline
		\multicolumn{1}{|c|}{\multirow{5}{*}{\begin{tabular}[c]{@{}c@{}}$1^{st}$\\ grid\end{tabular}}} & $\lambda_0^{(height)}$ & \begin{tabular}[c]{@{}c@{}}$(\Lambda_{D_1},u_\infty(\Sigma,D_1))$\\ $=(1.5,0.6,U_{0,0})$\end{tabular}       & \begin{tabular}[c]{@{}c@{}}$(\Lambda_{D_2},u_\infty(\Sigma,D_2))$\\ $=(1.5,0.8,U_{0,1})$\end{tabular}       & \begin{tabular}[c]{@{}c@{}}$(\Lambda_{D_3},u_\infty(\Sigma,D_3))$\\ $=(1.5,1,U_{0,2})$\end{tabular}       & \begin{tabular}[c]{@{}c@{}}$(\Lambda_{D_4},u_\infty(\Sigma,D_4))$\\ $=(1.5,1.2,U_{0,3})$\end{tabular}       \\ \cline{2-6}
		\multicolumn{1}{|c|}{}                                                                         & $\lambda_1^{(height)}$ & \begin{tabular}[c]{@{}c@{}}$(\Lambda_{D_5},u_\infty(\Sigma,D_5))$\\ $=(1.6,0.6,U_{1,0})$\end{tabular}       & \begin{tabular}[c]{@{}c@{}}$(\Lambda_{D_6},u_\infty(\Sigma,D_6))$\\ $=(1.6,0.8,U_{1,1})$\end{tabular}       & \begin{tabular}[c]{@{}c@{}}$(\Lambda_{D_7},u_\infty(\Sigma,D_7))$\\ $=(1.6,1,U_{1,2})$\end{tabular}       & \begin{tabular}[c]{@{}c@{}}$(\Lambda_{D_8},u_\infty(\Sigma,D_8))$\\ $=(1.6,1.2,U_{1,3})$\end{tabular}       \\ \cline{2-6}
		\multicolumn{1}{|c|}{}                                                                         & $\lambda_2^{(height)}$ & \begin{tabular}[c]{@{}c@{}}$(\Lambda_{D_9},u_\infty(\Sigma,D_9))$\\ $=(1.7,0.6,U_{2,0})$\end{tabular}       & \begin{tabular}[c]{@{}c@{}}$(\Lambda_{D_{10}},u_\infty(\Sigma,D_{10}))$\\ $=(1.7,0.8,U_{2,1})$\end{tabular} & \begin{tabular}[c]{@{}c@{}}$(\Lambda_{D_{11}},u_\infty(\Sigma,D_{11}))$\\ $=(1.7,1,U_{2,2})$\end{tabular} & \begin{tabular}[c]{@{}c@{}}$(\Lambda_{D_{12}},u_\infty(\Sigma,D_{12}))$\\ $=(1.7,1.2,U_{2,3})$\end{tabular} \\ \cline{2-6}
		\multicolumn{1}{|c|}{}                                                                         & $\lambda_3^{(height)}$ & \begin{tabular}[c]{@{}c@{}}$(\Lambda_{D_{13}},u_\infty(\Sigma,D_{13}))$\\ $=(1.8,0.6,U_{3,0})$\end{tabular} & \begin{tabular}[c]{@{}c@{}}$(\Lambda_{D_{14}},u_\infty(\Sigma,D_{14}))$\\ $=(1.8,0.8,U_{3,1})$\end{tabular} & \begin{tabular}[c]{@{}c@{}}$(\Lambda_{D_{15}},u_\infty(\Sigma,D_{15}))$\\ $=(1.8,1,U_{3,2})$\end{tabular} & \begin{tabular}[c]{@{}c@{}}$(\Lambda_{D_{16}},u_\infty(\Sigma,D_{16}))$\\ $=(1.8,1.2,U_{3,3})$\end{tabular} \\ \cline{2-6}
		\multicolumn{1}{|c|}{}                                                                         & $\lambda_4^{(height)}$ & \begin{tabular}[c]{@{}c@{}}$(\Lambda_{D_{17}},u_\infty(\Sigma,D_{17}))$\\ $=(1.9,0.6,U_{4,0})$\end{tabular} & \begin{tabular}[c]{@{}c@{}}$(\Lambda_{D_{18}},u_\infty(\Sigma,D_{18}))$\\ $=(1.9,0.8,U_{4,1})$\end{tabular} & \begin{tabular}[c]{@{}c@{}}$(\Lambda_{D_{19}},u_\infty(\Sigma,D_{19}))$\\ $=(1.9,1,U_{4,2})$\end{tabular} & \begin{tabular}[c]{@{}c@{}}$(\Lambda_{D_{20}},u_\infty(\Sigma,D_{20}))$\\ $=(1.9,1.2,U_{4,3})$\end{tabular} \\ \hline
	\end{tabular}
\bigskip
	\caption{Training dataset with the given characteristic grids $(\lambda^{(height)},\lambda^{(weight)})$. } \label{tab:training dataset}
\end{table}
\end{center}
Let $S_3(\Delta_j),j=1,\dots,M$ be a function subspace of $C^2([a_j,b_j])$ consisting of one dimensional, complex-valued functions in the direction of $\lambda^{(j)},j=1,\dots,M$ on the bounded interval $[a_j,b_j]$. The function in $S_3(\Delta_j),j=1,\dots,M$ is piecewise polynomial of degree 3 on every subinterval $[\lambda_{g_j-1}^{(l)},\lambda_{g_j}^{(j)}]$, where $g_j=1,\dots, k_j,j=1,\dots,M$. Then we introduce a {\it function subspace of multidimensional and complex-valued $C^2([a_j,b_j])$ functions} as
\begin{equation}\label{eq:function space of multidimensional}
    S_3\left(\Delta_1,\dots,\Delta_M\right),
\end{equation}
 on each rectangular grid
 \begin{equation}\label{eq:rectangular grid}
 	I_{g_1, g_2, \dots, g_{M}} := \prod_{j\in\mathscr{C}}[\lambda^{(j)}_{g_j},\lambda^{(j)}_{g_j+1}]
 \end{equation}
 for all $0\leq g_j\leq k_j-1,j=1,\dots,M$ that are piecewise polynomials of degree $3$ on every interval. For easy reference we provide the definition of B-splines.

\begin{defn}\label{defn:B-spline}
	The sets of $k_j+k,j=1,\dots,M$, B-spline basis functions $\{B_{l,k}(\lambda^{(j)})\}_{l=1}^{k_j+k}$ of degree $k$ of the function space $S_k(\Delta_j)$ are defined based on concurrent boundary knots vectors with Cox-deBoor recurrence\cite{de1972calculating},
	\begin{equation}\label{eq:recursion}
	B_{l,k}(\lambda^{(j)})=\frac{\lambda^{(j)}-\lambda^{(j)}_{l-1-k}}{\lambda^{(j)}_{l-1}-\lambda^{(j)}_{l-1-k}}B_{l-1,k-1}(\lambda^{(j)})+\frac{\lambda^{(j)}_{l}-\lambda^{(j)}}{\lambda^{(j)}_{l}-\lambda^{(j)}_{l-k}}	B_{l,k-1}(\lambda^{(j)})
	\end{equation}
	with
	\begin{equation}
	B_{l,0}(\lambda^{(j)})=\begin{cases}
	1&\text{if } \lambda^{(j)}_{l}\leq \lambda^{(j)}<\lambda^{(j)}_{l+1}\\
	0& \text{else}
	\end{cases},
	\end{equation}
	for $l=1,\dots, k_j+k$, where $k=3$ is for the cubic B-spline and $\lambda_l^{(j)}$ are elements of the knot vectors, satisfying the relation $\lambda_l^{(j)} < \lambda_{l+1}^{(j)}$.
\end{defn}

All methods are in the following using splines with a $k + 1$ regular knot vector, and the interior knots are the grid points. Based on Definition~\ref{defn:B-spline}, we next introduce a general learning model for the geometric shape generation through the multidimensional cubic B-spline interpolation.

\medskip

\noindent{\bf Body Learning Model I.}~~Given the training dataset $\mathbf{Z}:=\{(\Lambda_{D_i},u_\infty(\Sigma;D_i))\}_{i\in \mathscr{N}}$, the learning model $T_\mathbf{Z}\in S_3(\Delta_1,\dots,\Delta_M)$ at $\Lambda=(\lambda^{(1)},\dots,\lambda^{(M)})$ for the geometric body generation associated with the sets of the grids $\{\Delta_j\}_ {j\in \mathscr{C}}$ is defined as follows
	\begin{equation}\label{eq:cubic spline}
	T_\mathbf{Z}(\Lambda)=\sum_{g_1=1}^{ k_1+3}\cdots\sum_{g_M=1}^{k_M+3}c_{g_1,g_2,\dots ,g_{M}} \prod_{j\in\mathscr{C}}B_{g_j,3}^j(\lambda^{(j)}),\quad\lambda^{(j)}\in[a_j,b_j],
	\end{equation}
which satisfies the following conditions by \eqref{eq:traininng model}
	\begin{equation}\label{eq:input-output relation nonuniform}
	T_\mathbf{Z}(\Lambda_{D_i})=\sum_{g_1=1}^{ k_1+3}\cdots\sum_{g_M=1}^{k_M+3}c_{g_1,g_2,\dots, g_{M}} \prod_{j\in\mathscr{C}}B_{g_j,3}^j(\lambda^{(j)}_{i})={u}_\infty(\Sigma;D_{i}),
	\end{equation}
	where $c_{g_1,g_2,\dots, g_{M}}$ with $g_j=0,\dots,k_j,j=1,\dots,M$ are the coefficients to be determined from the training dataset $\mathbf{Z}$, $B_{g_j,3}^j(\lambda^{(j)}), j\in \mathscr{C}$ are the B-spline basis functions of degree 3 defined in \eqref{eq:recursion}, and $k_j,j=1,\dots,M$ is the number of different characteristic values in each direction.

\medskip

\begin{rem}
In Learning Model I, \eqref{eq:cubic spline} presents a general form of the learning model $T_\mathbf{Z}$ for the geometric body generation associated with the non-uniform grids $\{\Delta_j\}_ {j\in \mathscr{C}}$. The learning model eventually generates a B-spline interpolation with the associated spacing for each segment. If the training dataset consists of equidistant grids, we can derive a faster and easier learning model and this shall be provided in the next subsection.
\end{rem}

\subsection{Uniform B-spline}\label{Uniform B-spline}
In this subsection, we derive a learning model for a special case with the training dataset consisting of equidistant grids. By \eqref{eq:grids}, for the $M$ set of equidistant grids $	\Delta_j=\{\lambda_0^{(j)},\cdots,\lambda_{k_j}^{(j)}\}_{j\in \mathscr{C}}$  with additional conditions
\begin{equation}\label{eq:equidistant knots}
\lambda_{g_j}^{(j)}=a_j+g_j h_j,\quad h_j=\frac{b_j-a_j}{k_j},\quad, g_j=1,\dots, k_j,
\end{equation}
the B-spline basis function $\beta^k(t)$ of degree $k$ is a symmetrical, bell-shaped function constructed from $k+1$ times self-convolution of the $\beta^0(t)$ basis function of degree zero which is a centered rectangle around origin \cite{schoenberg1946contributions}
\begin{equation}\label{eq:beta zero}
\beta^0(t)=\begin{cases}
1,&-\frac{1}{2}<t<\frac{1}{2}\\
\frac{1}{2},&|x|=\frac{1}{2}\\
0,&\text{otherwise},
\end{cases}
\end{equation}
\begin{equation}\label{eq:self-convolution}
\beta^k(t)=\underbrace{\beta^0(t)\star\dots\star\beta^0(t)}_\text{$(k+1)$ times}.
\end{equation}
The centered symmetric B-spline of degree $k$ has an explicit expression \cite{schoenberg1973cardinal}
\begin{equation}\label{eq:explicit expression of beta n}
\beta^k(t)=\frac{1}{k!}\sum^{k+1}_{j=0}C^{k+1}_j(-1)^j(t+\frac{k+1}{2}-j)^k_+,
\end{equation}
where the function $x_+$ is defined as follows
\begin{equation}
x_+=\begin{cases}
x,&\text{for}\quad x>0,\\
0,&\text{otherwise}.
\end{cases}
\end{equation}
In this paper, we are particular intereted in the cubic B-spline. By \eqref{eq:explicit expression of beta n}, the closed-form representation of the cubic B-spline basis function can be also expressed as
\begin{equation}\label{eq:cubic basis}
\beta^3(t)=\frac{1}{6}\begin{cases}
(2-|t|)^3 & 1<|t| \leq 2,\\
4-6|t|^2+3|t|^3, & |t|\leq 1,\\
0, & \text{elsewhere},
\end{cases}
\end{equation}
which is used for preforming the interpolation. Then we choose the interpolation kernels to be
\begin{equation}\label{eq:uniform basis}
L_{g_j}(\lambda^{(j)})=\beta^3 \left(\lambda^{(j)}\right), \quad k_{j}=1,\dots, N_{pair}+3,
\end{equation}
as the basis of $S_3(\Delta_j)$ in the $\lambda^{(j)}$-direction such that
$L=\{L_1,L_2,\dots,L_{k_j+3}\}$
is a basis of the $(k_j+3)$-dimensional space $S_3(\Delta)$
and hence, the basis of the $\prod_{j\in\mathscr{C}}(k_j+3)$-dimensional space $S_3(\Delta_1,\dots,\Delta_M)$ in the directions of $\lambda^{(1)},\dots,\lambda^{(M)}$ is given by
\begin{equation}\label{eq:multi uniform basis}
	\{L_{g_1}L_{g_2}\dots L_{g_M}\vert g_j\in \{1,\dots, N_{pair}+3 \},j\in{\mathscr{C}}\}.
\end{equation}
Based on \eqref{eq:multi uniform basis},\eqref{eq:grids} and \eqref{eq:equidistant knots}, we next introduce the learning model for the uniform case.

\medskip

\noindent{\bf Body Learning Model II.}~~
	Given the training dataset $\mathbf{Z}:=\{(\Lambda_{D_i},u_\infty(\Sigma;D_i))\}_{i\in \mathscr{N}}$, the learning model $T_\mathbf{Z}\in S_3(\Delta_1,\dots,\Delta_M)$ at $\Lambda=(\lambda^{(1)},\dots,\lambda^{(M)})$ for the geometric body generation associated with the sets of equidistent grids $\{\Delta_j\}_ {j\in \mathscr{C}}$ defined in \eqref{eq:equidistant knots} is defined as follows

	\begin{equation}\label{eq:cubic spline uniform}
	T_\mathbf{Z}(\Lambda)=\sum_{g_1=1}^{ k_1+3}\cdots\sum_{g_M=1}^{k_M+3}c_{g_1,g_2,\dots, g_{M}} \prod_{j\in {\mathscr{C}}}L_{g_j}^j(\lambda^{(j)}),
	\end{equation}
which is required to satisfy the following conditions
	\begin{equation}\label{eq:input-output relation uniform}
	T_\mathbf{Z}(\Lambda_{D_i})=\sum_{g_1=1}^{ k_1+3}\cdots\sum_{g_M=1}^{k_M+3}c_{g_1,g_2,\dots ,g_{M}} \prod_{j\in {\mathscr{C}}}L_{g_j}^j(\lambda^{(j)}_{i})={u}_\infty(\Sigma;D_{i}),
	\end{equation}
	where $c_{g_1,g_2,\dots ,g_{M}}$ with $g_j=0,\dots,k_j,j=1,\dots,M.$ are the coefficients to be determined from the training dataset $\mathbf{Z}$, $L_{g_j}^j(\lambda^{(j)})$ are B-spline basis functions of degree 3 defined in \eqref{eq:uniform basis} and $k_j,j=1,\dots,M$ is the number of different characteristic values in each direction.

\medskip

\subsection{Natural Spline}\label{Minimization problem}

In Learning Models I and II, the learning functionals for the non-uniform and uniform case are respectively considered in the $\prod_{j\in \mathscr{C}}(k_j+3)$-dimensional space $S_3(\Delta_1,\dots,\Delta_M)$. $\prod_{j\in \mathscr{C}}(k_j+3)$ interpolation conditions are required to determine the coefficients $c_{i_1,i_2,\dots, i_{M}}$ in the training models. However, there are only $\prod_{j\in \mathscr{C}}(k_j+1)$ shape generators to specify $\prod_{j\in \mathscr{C}}(k_j+1)$ conditions in \eqref{eq:input-output relation nonuniform} or \eqref{eq:input-output relation uniform}. To obtain a unique correlation between the characteristic values and the shape generator, we need to add $\prod_{j\in \mathscr{C}}2=2^{M}$ conditions, which define the second-order derivatives of the spline function at the boundary $a_j$ and $b_j$ to be equal to 0 and lead to a natural spline.

\subsection{Prediction on shape generator}

With the Learning Models I and II established in the previous subsections, for an input new set of characteristic values $\Lambda_{D_{new}}$ associated with a new geometric body $D_{new}$, the unknown shape generator can be generated as follows,
\begin{equation}\label{eq:predict far field pattern}
T_\mathbf{Z}(\Lambda_{D_{{new}}})=\sum_{g_1=1}^{ k_1+3}\cdots\sum_{g_M=1}^{k_M+3}c_{g_1,g_2,\dots, g_{M}} \prod_{j\in {\mathscr{C}}}L_{g_j}^j(\lambda^{(j)}_{{new}})=\hat{u}_\infty(\Sigma;D_{{{new}}})\approx u_\infty(\Sigma;D_{{{new}}}),
\end{equation}
where the coefficients $c_{g_1,g_2,\dots, g_{M}}$ in \eqref{eq:predict far field pattern} could be determined by solving the natural spline problem.

\subsection{Reconstruction}\label{Reconstruction}
In this subsection, we briefly outline the Fourier method for the reconstruction of geometry shape $D_{{new}}$ by using the  shape generators ${u}_\infty(\Sigma;D_{new})$.

Define the periodic  Sobolev space by
\begin{equation*}
  H^{\sigma}(\mathbb{R}^d)=\left \{ g\in L^2(\mathbb{R}^d): (1+| \xi|^2)^{\frac{\sigma}{2}}  \widehat{g}(\xi) \in L^2(\mathbb{R}^d),\right \},
\end{equation*}
where $\sigma \geq 1$, $ \xi\in \mathbb{Z}^d$ and $\widehat g(\xi)$ denote the Fourier coefficients of $g$.
Suppose that $f\in H^{\sigma}(\mathbb{R}^d)$ has a compact support in domain $V_0=(-a/2, a/2)^d,\, (a>0)$,
then the Fourier transform of $f$ is represented by
\begin{equation}\label{eq:Fourier_coefficient}
   \widehat{f}_{\xi}=\frac{1}{a^d} \int_{V_0}f(x)\,\overline{\phi_{ \xi }(x)}\,\mathrm{d}x,
\end{equation}
where the overbar stands for the complex conjugate and the Fourier basis functions are given by
\begin{equation*}
  \phi_{ \xi}(x)=\exp\left(\mathrm{i}\frac{2\pi}{a}  \xi \cdot x \right),\quad \xi\in \mathbb{Z}^d.
\end{equation*}

\begin{defn}[Admissible wavenumbers and observation directions]\label{def:wavenumber}
Let $\mu$ be a sufficiently small positive constant such that $0<\mu<1$ and
\begin{equation*}
   \xi_{0}:=
   \begin{cases}
   (\mu,0), & d=2,\medskip\\
   (\mu,0,0), & d=3,\\
   \end{cases}
\end{equation*}
then the admissible set of wavenumbers is defined by
\begin{equation*}
 \mathbb{K}:= \left\{\frac{2\pi}{a}|{\xi}|:{\xi}\in \mathbb{Z}^{3}\backslash \{ 0\}  \right\} \cup k_0,
\end{equation*}
correspondingly, the admissiable set of observation directions is given by
\begin{equation*}
 \mathbb{X}:= \left\{\frac{ \xi}{|{\xi}|}:{\xi}\in \mathbb{Z}^{3}\backslash \{ 0\} \right\} \cup \hat x_0,
\end{equation*}
where $k_0=2\pi|\xi_0|/a$ and $\hat x_0= \xi_0/|\xi_0|$ for $\xi=0\in \mathbb{Z}^d$.
\end{defn}

Due to $\mathrm{supp}\, S \subset\subset V_0$, for $ \xi\in \mathbb{Z}^{3}\backslash \{0\}$, the far-field pattern defined in
\eqref{eq:sss1} can be written as
\begin{equation}\label{eq:far2}
	u_\infty(\hat x, k; f)
      = C_{d,k}\int_{\mathbb{R}^d} e^{-ik\hat{x}\cdot y}
  f(y)\ dy\\
     = C_{d,k}\int_{V_0} e^{-ik\hat{x}\cdot y}
  f(y)\ dy,
\end{equation}
where $k\in \mathbb{K}$ and $\hat x\in \mathbb{X}$ depend on $\xi$.
Combining \eqref{eq:Fourier_coefficient} and \eqref{eq:far2}, one has

\begin{equation}\label{eq:f_coefficients}
\widehat{f}_{\xi}=\frac{1}{a^d C_{d,k}}u_\infty(\hat x, k; f),\quad \xi\in \mathbb{Z}^{3}\backslash \{ 0\}.
\end{equation}
For $\xi=  0$,  using the Fourier expansion of $f$,
 we derive that
\begin{equation*}
\begin{aligned}
	u_\infty(\hat x_0, k_0; f)
     &=C_{d,k}\int_{V_0} e^{-ik_0\hat{x}_0\cdot y}
  f(y)\ dy,\\
     &= C_{d,k} \int_{V_0}\left( \widehat{f}_{0}+\sum_{\xi \in \mathbb{Z}^d \backslash\{ 0\}} \widehat{f}_{ \xi}\,  \phi_{ \xi}  \right)\overline{\phi_{\xi_0}(y)} \ dy\\
     &= C_{d,k} \int_{V_0} \widehat{f}_{0} \overline{\phi_{\xi_0}(y)} \ dy +C_{d,k}  \sum_{\xi \in \mathbb{Z}^d \backslash\{ 0\}} \widehat{f}_{\xi} \int_{V_0} \phi_{\xi}(y) \overline{\phi_{\xi_0}(y)} \  dy\\
     &= a^d C_{d,k}\frac{\mathrm{\sin}(\mu \pi)}{\mu \pi}\widehat{f}_{ 0}
     +C_{d,k} \sum_{\xi \in \mathbb{Z}^d \backslash\{0\}} \widehat{f}_{ \xi} \int_{V_0} \phi_{\xi}(y) \overline{\phi_{\xi_0}(y)} \ dy,
\end{aligned}
\end{equation*}
which implies
\begin{equation}\label{eq:f0_coefficients}
  \widehat{f}_{0}= \frac{\mu \pi}{a^d \sin\mu \pi}\left\{  \frac{u_{\infty}(\hat{x}_{ 0 },k_{ 0}; f )}{C_{d,k}} -  \sum_{{\xi}\in \mathbb{Z}^{3}\backslash \{0\}} \widehat{f}_{ {\xi}} \int_{V_0} \phi_{ {\xi}}(y) \overline{\phi_{ {\xi}_0}(y)} \ dy    \right\}.
\end{equation}
Therefore, the Fourier method is to approximate $f$ by a truncated Fourier expansion
 \begin{equation*}
    f_N=\widehat{f}_{0}+ \sum_{1\leq |\xi|_{\infty}\leq N} \widehat{f}_{ \xi}\,  \phi_{ \xi},
 \end{equation*}
where $N\in \mathbb{N}_+$ denotes the truncation order and the Fourier coefficients are given by \eqref{eq:f_coefficients} and \eqref{eq:f0_coefficients}. Hence  the domain $D_{new}$ is determined  since the set $D_{new}$ is the external \emph{shape} of $f$.

Next, we investigate the stability of the proposed Fourier method.
In practical computation, there exists some noise between the shape generators $u_\infty(\Sigma;D_{{new}})$  and the predictions $\hat u_\infty(\Sigma;D_{{new}})$, which satisfies
\begin{equation*}
  \left\|\hat u_\infty(\Sigma;D_{{new}})- u_\infty(\Sigma;D_{{new}})\right\|_{L^2}\leq {\delta}\| u_\infty(\Sigma;D_{{new}})\|_{L^2},
\end{equation*}
where $\delta >0$ denotes the noise level. Noting that $(\hat x, k)\in \mathbb{X}\times \mathbb{K}\subset \Sigma$, then the approximation of $f$ from  predicted  shape generators is given by
 \begin{equation*}
    f_N^{\delta}=\widehat{f}_{0}^{\delta}+ \sum_{1\leq |\xi|_{\infty}\leq N} \widehat{f}_{ \xi}^{\delta} \,  \phi_{ \xi},
 \end{equation*}
where
\begin{align}
&\label{eq:s_coefficients_noise}\widehat{f}^{\delta}_{\xi}=\frac{1}{a^d C_{d,k}}\hat u_\infty(\hat x, k),\quad \xi\in \mathbb{Z}^{3}\backslash \{ 0\},\\
 &\label{eq:s0_coefficients_noise} \widehat{f}_{0}^{\delta}= \frac{\mu \pi}{a^d \sin\mu \pi}\left\{  \frac{\hat u_{\infty}(\hat{x}_{ 0 },k_{ 0} )}{C_{d,k}} -  \sum_{{\xi}\in \mathbb{Z}^{3}\backslash \{0\}} \widehat{f}_{ {\xi}}^{\delta} \int_{V_0} \phi_{ {\xi}}(y) \overline{\phi_{ {\xi}_0}(y)} \ dy    \right\}.
  \end{align}

\begin{thm}\label{thm:stability}
Let  $f$ be a compactly supported function in $H^{\sigma}(\mathbb{R}^d)$, $\sigma \geq 1$, with $ \mathrm{supp} f \subset\subset V_0$,  then we have the following estimate
\begin{equation*}
  \|f^{\delta}-f\|_{L^2(V_0)}^2\leq C\delta+ C (\tau^d+\tau^{-2}) \delta^{\frac{4}{2+d}}.
\end{equation*}
where $d\leq \tau\in \mathbb{R}_+$ and $C$ is a constant which depends on $f, a, d, \mu$.
\end{thm}
\begin{proof}[\bf Proof.]
Using the Plancherel theorem, we have
\begin{equation}\label{eq:main}
\begin{aligned}
  \|f^{\delta}-f\|_{L^2(D)}^2
  &=\|f^{\delta}-f\|_{L^2(\mathbb{R}^d)}^2\\
  &=\frac{1}{a^{2d}}\int_{\mathbb{R}^d} |\widehat f_{\xi}^{\delta}- \widehat f_{\xi}|^2 \   d\xi\\
  &=\frac{1}{a^{2d}}\left(\int_{|\xi|_{\infty}\leq N} |\widehat f_{ \xi}^{\delta}- \widehat f_{\xi}|^2 \  d\xi + \int_{|\xi|_{\infty}> N} |\widehat f_{\xi}^{\delta}- \widehat f_{\xi}|^2 \ d\xi   \right),
  \end{aligned}
\end{equation}
where $N\in \mathbb{N}_+$.
Due to $f\in H^{\sigma}(\mathbb{R}^d)$, that is,
\begin{equation*}
  \left(1+|\xi|^2\right)^{\frac{\alpha}{2}} \widehat f_{\xi} \in L^2(\mathbb{R}^d), \quad \forall\, |\alpha|\leq \sigma.
\end{equation*}
It means that both $|\xi|\widehat f_{\xi}$ and $|\xi|\widehat f_{\xi}^{\delta}$ are bounded in $L^2(\mathbb{R}^d)$, so we can find $N>0$ , such that
\begin{equation}\label{eq:part2}
 \int_{|\xi|_{\infty}> N} |\widehat f_{\xi}^{\delta}- \widehat f_{\xi}|^2 \ d\xi\leq \frac{1}{N^2}\int_{|\bm \xi|_{\infty}> N}
 |\xi|^2|\widehat f_{\xi}^{\delta}- \widehat f_{\xi}|^2 \ d\xi<\frac{C_1}{N^2},
\end{equation}
where $C_1>0$ is a constant.
For $1\leq |\xi|_{\infty}\leq N$, from \eqref{eq:f_coefficients} and \eqref{eq:s_coefficients_noise}, we have
\begin{equation*}
\begin{aligned}
  |\widehat f_{\xi}^{\delta}- \widehat f_{\xi}|
  &=\frac{1}{a^{d} C_{d,k}}\left|\hat u_{\infty}(\hat x, k)-u_{\infty}(\hat x, k)\right|\\
  &\leq \frac{\delta}{a^{d} C_{d,k}}|u_{\infty}(\hat x, k)|\\
  &=\frac{\delta}{a^{d} C_{d,k}}\left|C_{d,k}\int_{V_0} f(y) e^{-ik\hat{x}\cdot y}\ dy\right|\\
  &\leq \frac{\delta}{a^{d}}\left(\int_{V_0} |f(y)|^2 \ dy\right)^{\frac{1}{2}}
  \left(\int_{V_0}|e^{-ik\hat{x}\cdot y}|^2\,\mathrm{d}y\right)^{\frac{1}{2}}\\
  & = \frac{\|f\|_{L^2(V_0)}}{a^{d/2}}\delta ,
   \end{aligned}
\end{equation*}
which implies
\begin{equation}\label{eq:part1}
  \int_{1\leq|\xi|_{\infty}\leq N} |\widehat f_{\xi}^{\delta}- \widehat f_{ \xi}|^2 \ d\xi\leq C_2 (2N+1)^d \delta^2,
\end{equation}
for $C_2=\|f\|_{L^2(V_0)}^2/ a^d$.
Define $ \xi=(\xi_{1},\xi_2) \in \mathbb{Z}^2$ or $ \xi=(\xi_{1},\xi_2, \xi_3) \in \mathbb{Z}^3$, by a straight forward calculation, one finds that
\begin{equation*}
\int_{V_0}\phi_{\xi}(y) \overline{\phi_{\xi_0}(y)}\ dy =\left\{
\begin{aligned}
 &0,   & | {\xi}|\neq|\xi_1|,  \\
 &-\frac{a^d\cos \xi_1\pi \sin\lambda\pi}{(\xi_{1}-\lambda)\pi},  & | {\xi}|=|\xi_1|.
\end{aligned}
\right.
\end{equation*}
For $\xi =0$, using \eqref{eq:f0_coefficients}, \eqref{eq:s0_coefficients_noise},
 \eqref{eq:part2}, \eqref{eq:part1} and the last equation,  it derives that
\begin{equation}\label{eq:part0}
\begin{aligned}
  |\widehat f_{ 0}^{\delta}- \widehat f_{ 0}|
   \leq & \frac{\mu \pi}{ a^d C_{d,k}   \sin \mu \pi} \left| \hat u_{\infty}^{\delta}(\hat{x}_{0},k_{0})-u_{\infty}^{\delta}(\hat{x}_{0},k_{0})
   \right| \\
   & +\frac{\mu \pi}{a^{d} \sin \mu \pi}\sum_{1\leq| {\xi}|_{\infty}\leq N}\left|\widehat{f}^{\delta}_{\xi}- \widehat{f}_{ {\xi}}\right|  \left|\int_{V_0}\phi_{\xi}(y) \overline{\phi_{ {\xi}_0}(y)}\ dy \right|\\
   &+\frac{\mu \pi}{a^{d}\sin \mu \pi} \sum_{|{\xi}|_{\infty}\geq N}\left|\widehat{f}_{\xi}^{\delta}-\widehat{f}_{ \xi}\right| \left|\int_{V_0}\phi_{\xi}(x) \overline{\phi_{ \xi_0}(x)}\ dy \right|\\
  \leq & C_3\delta+ \sqrt{C_2} (2N+1)^d \delta+\frac{\sqrt{C_1}}{N},
   \end{aligned}
\end{equation}
where $C_3={\mu \pi\|f\|_{L^2(V_0)} }/(a^{d/2}\sin \mu \pi)$.
Hence, substituting \eqref{eq:part2}, \eqref{eq:part1} and \eqref{eq:part0} into \eqref{eq:main}, it deduces that
\begin{equation*}
   \|f^{\delta}-f\|_{L^2(V_0)}^2\leq C\delta^2+ C N^d \delta^2+\frac{C}{N^2},
\end{equation*}
where $C=\max\{2 C_1, 2^{d+1}C_2, C_3^2\}/ a^{2d}$.
Furthermore, if we take $N=\tau\delta^{-\frac{2}{2+d}}$ with $\tau\geq d$ in Theorem \ref{thm:stability}, then it holds that
\begin{equation*}
  \|f^{\delta}-f\|_{L^2(V_0)}^2\leq C\delta+ C (\tau^d+\tau^{-2}) \delta^{\frac{4}{2+d}}.
\end{equation*}
\end{proof}

Let $N=\left[\tau\delta^{-\frac{2}{2+d}},\, \tau\geq d\right]$, here $[X]$ denotes the largest integer that is smaller than $X + 1$.
From definition \ref{def:wavenumber},  the truncated wavenumbers and observation directions can be written as
\begin{equation*}
\begin{aligned}
& \mathbb{K}_{N}:= \left\{\frac{2\pi}{a}|{\xi}|:1\leq |{\xi}|\leq N  \right\} \cup k_0,\\
& \mathbb{X}_{N}:= \left\{\frac{ \xi}{|{\xi}|}:1\leq |{\xi}|\leq N \right\} \cup \hat x_0.
\end{aligned}
\end{equation*}
Thus, the truncated Fourier expansion of $f$ from the predictions $\{\hat u_\infty(\Sigma;D_{{new}_q})\}_{q\in\mathscr{Q}}$  takes the form
  \begin{equation}\label{Fourier expansion}
    f_{N}^{\delta}:=\widehat{f}_{0}^{\delta} + \sum_{1\leq |\xi|_{\infty}\leq N} \widehat{f}_{\xi}^{\delta}\,  \phi_{\xi}(x),
 \end{equation}
where
\begin{equation}\label{Fourier expansion_coefficients}
    \begin{aligned}
&\widehat{f}^{\delta}_{\xi}=\frac{1}{a^d C_{d,k}}\hat u_\infty(\hat x, k),\quad 1\leq |\xi|_{\infty}\leq N,\\
& \widehat{f}_{0}^{\delta}= \frac{\mu \pi}{a^d \sin\mu \pi}\left\{  \frac{\hat u_{\infty}(\hat{x}_{ 0 },k_{ 0} )}{C_{d,k}} -  \sum_{1\leq |\xi|_{\infty}\leq N} \widehat{f}_{ {\xi}}^{\delta} \int_{V_0} \phi_{ {\xi}}(y) \overline{\phi_{ {\xi}_0}(y)} \ dy    \right\}.
  \end{aligned}
\end{equation}
%

\subsection{Summary}
Motivated by above discussion, we are ready to present our novel modeling methodology for geometric shape in $\mathbb{R}^d, d=2,3$, see Algorithm \ref{GDG}.
\begin{algorithm}\label{Algorithm GDG}
	\caption{Inverse-scattering-based geometric body generation scheme}
	\label{GDG}
	\begin{algorithmic}[1]
\State Select the parameter $N$, the set of admissible wavenumbers $ \mathbb{K}_{N}$ and the set of admissible observation directions  $\mathbb{X}_{N}$.
\State 	Given a training dataset $\mathbf{Z}:=\{(\Lambda_{D_i},u_\infty(\hat x, k;D_i))\}_{i\in \mathscr{N}}$ for $\hat x\in X_N, k\in K_N $, obtain the coefficients $c_{g_1,\dots,g_{M}}$ of the learning model $T_\mathbf{Z}$ by solving the problem of the natural spline interpolation.
\State Given the characteristic sets $\Lambda_{D_{new}}$, predict the new shape generators $\{\hat u_\infty(\hat x, k)\}$ for $\hat x\in X_N, k\in K_N $ with the use of the learning model $T_\mathbf{Z}$.
\State Compute the Fourier coefficients $\widehat{f}_{ 0}^{\delta}$ and $\widehat{f}_{ \xi}^{\delta}$ defined in \eqref{Fourier expansion_coefficients} for $1\leq| \xi|_\infty\leq N_t$.
\State Select a sampling mesh $\mathcal{T}_h$ in a region $V_0$. For each sampling point $z_j\in \mathcal{T}$, calculate the imaging function $f_{N}^{\delta}$ defined in \eqref{Fourier expansion}.  $D_{new}$ is obtained as the external \emph{shape} of $f_{N}^{\delta}$.	
	\end{algorithmic}
\end{algorithm}

\section{Numerical examples }

In this section, several  numerical examples are conducted to show that the proposed method is effective and efficient.

The proposed algorithm is implemented by using Matlab 2016. The shape generator $\{u_\infty(\hat x, k;D_i)\}_{i\in \mathscr{N}}$ is obtained by solving the direct problem of \eqref{sourceScattering}. To avoid the inverse crime, we use the quadratic finite elements on a truncated spherical domain enclosed by a PML layer. The mesh of the forward solver is successively refined till the relative error of the successive measured scattered data is below $0.1\%$. Then artificial shape generators are generated by applying the Kirchhoff integral formula to the scattered data.
Thus the training dataset is given by
\begin{equation*}
\left\{ (\Lambda_{D_i},u_\infty(\hat x_j, k_j;D_i)): \hat{x}_{j}\in \mathbb{X}_{N},\, k_{j}\in \mathbb{K}_{N}\right \},
\end{equation*}
where  $j=1,2, \cdots,(2N+1)^{d}$ and $i=1,2,\cdots, N_t$ denotes the $i$-th geometry shape. In what follows, we set $\tau=2$ ($\tau=3$) for $d=2$ ($d=3$) and $\delta=1\%$, then we have $N=20$ ($N=19$) for $d=2$ ($d=3$).

Next, we present the implementation of interpolation.
The characteristic value set is given by $\{\Lambda_{D_i}\}$, where
$\Lambda_{D_i}=\{\lambda_i^{(1)}, \lambda_i^{(2)},\cdots,\lambda_i^{({M})}\}$ has ${M}$ variables. For a fixed wavenumber $k_j$, we use cubic spline interpolation
to obtain the coefficients $c_{i_1\dots,i_{{M}}}$  from the characteristic value $\Lambda_{D_i}$ and the shape generator $u_{\infty}(\hat{x}_j,k_j;D_i)$. Therefore, given a new
characteristic value $\Lambda_{D_{new}}=\{\lambda_{D_{new}}^{(1)}, \lambda_{D_{new}}^{(2)},\cdots,\lambda_{D_{new}}^{(M)}\}$, we  obtain the predicted shape generators $\{\hat u_{\infty}(\hat{x}_j,k_j; D_{new})\},\, j=1,2,\cdots, (2N+1)^d$.

Finally, we specify details of recovering the geometry. As discussed above, reconstructing the geometry shape is equally to reconstructing the source function $f$.
In the discrete formula, the domain $V_0$ is divided into a uniform mesh with size $100\times 100$ in two dimensions and  size $100\times 100\times 100 $ in three dimensions.
Further, the approximated  Fourier series $f_{N}^{\delta}$ are computed at the mesh nodes $\mathcal{T}_j,\, j=1, \cdots, 100^3$ in \eqref{Fourier expansion}.
Thus, the  geometry shape $D_{new}$ is approximated by the boundary of the imaging results $f_{N}^{\delta}$.

\subsection{Kite shaped Experiments}
In the first example,  we  aim to  reconstruct a kite shaped domain with  scale changing.
The kite shaped domain is parameterized by
\begin{equation*}
x(t)=( \beta_1(\cos t +0.65\cos 2t -0.65), \ 1.5\beta_2\sin t), \quad t\in [0, 2\pi],
\end{equation*}
where $\beta_1$ and $\beta_2$ are scale factors (characteristic values)  with $\beta_1, \beta_2 \in [0.5,\,1.8]$.
The training dataset consists of $14 \times 14$ different scale domains, i.e., $\beta_1$ and $\beta_2$ uniformly distributed on $[0.5,\,1.8]$ with $M=14^2$.
Next, we consider four sets of different scale factors which are not covered by the training data. In this numerical experiments, the imaging results with different characteristic values are shown in Figure \ref{fig:kite}, where the black dotted lines denote the exact boundary. It is clear that the reconstructions are very closed to the exact domains.

\begin{figure}
	\subfigure[]{\includegraphics[width=0.48\textwidth]{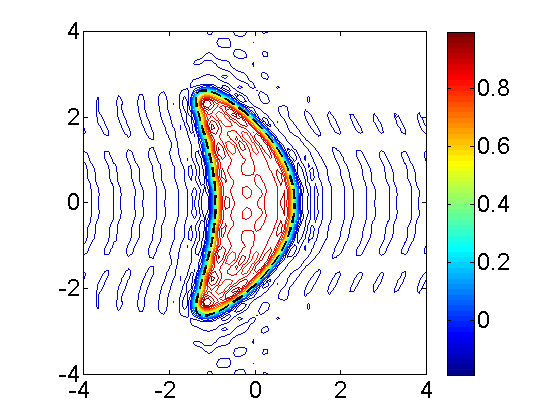}}
	\subfigure[]{\includegraphics[width=0.48\textwidth]{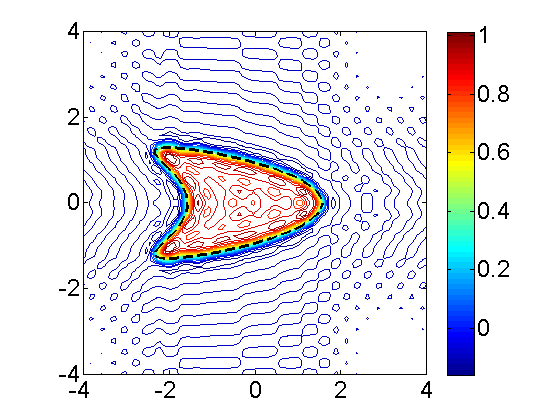}}\\
	\subfigure[]{\includegraphics[width=0.48\textwidth]{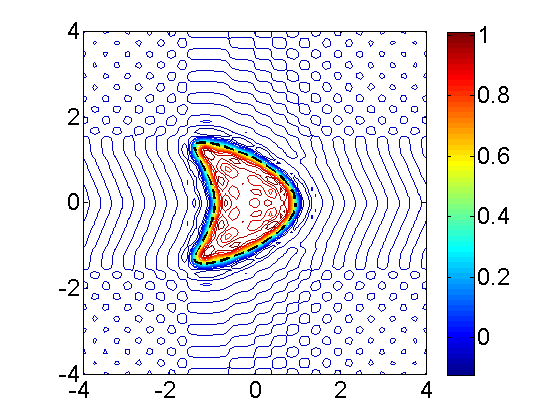}}
	\subfigure[]{\includegraphics[width=0.48\textwidth]{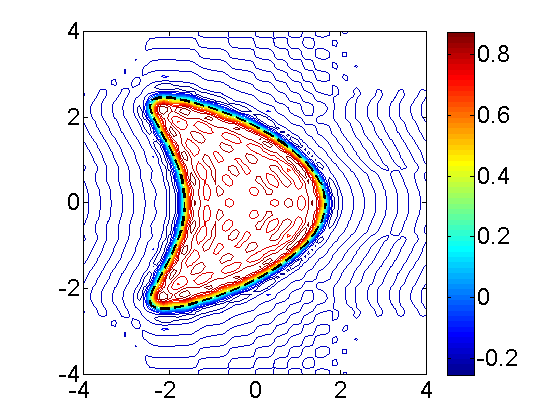}}
	\caption{\label{fig:kite} Contour plots of reconstructed kite shape with different scale factors $(\beta_1, \beta_2)$:
		(a) $(0.93, 1.76)$, (b) $(1.58, 0.87)$, (c) $(0.95, 0.95)$, (d) $(1.65, 1.65)$. }
\end{figure}

\subsection{Rounded triangle and apple shaped Experiments}
In the second example,  we  aim to  recover multi-domain with different scale factors.
The apple shaped domain is parameterized by
\begin{equation*}
y(t)= \beta_1((0.5+0.4\cos t+0.1 \sin 2t)/(1+0.7\cos t))(\cos t,\  \sin t) \quad t\in [0, 2\pi],
\end{equation*}
and the rounded triangle shaped domain is parameterized by
\begin{equation*}
z(t)= \beta_2(1+0.15\cos 3t)(\cos t,\  \sin t) \quad t\in [0, 2\pi],
\end{equation*}
where $\beta_1\in [1, 2]$ and $\beta_2\in [0.5, 1.5]$ are scale factors (characteristic values) for different domains.
The training dataset consists of $11 \times 11$ different scale domains, that is, $\beta_1$ and $\beta_2$ uniformly distributed on $[1, 2 ]\times [0.5, 1]$.
Similarly, we give four sets of different scale factors which are not covered by the training data.  Figure \ref{fig:mix} shows the
the reconstruction of multi-domain with different characteristic values via contour plots,
where the black dotted lines denote the exact boundary. It demonstrates very good imaging performance of the approach.

\begin{figure}
	\subfigure[]{\includegraphics[width=0.48\textwidth]{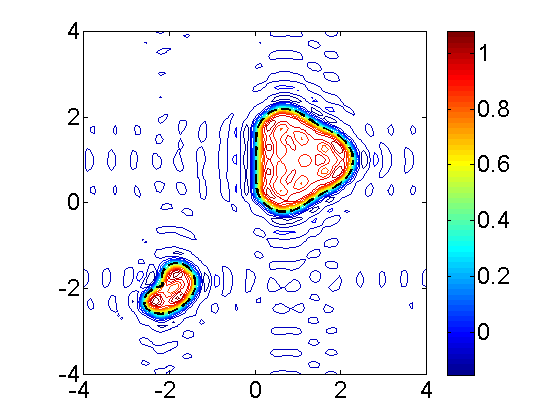}}
	\subfigure[]{\includegraphics[width=0.48\textwidth]{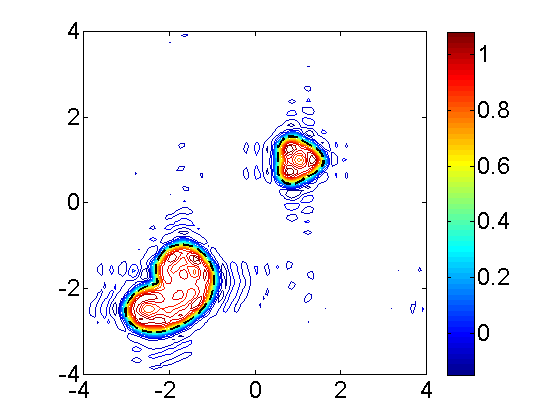}}\\
	\subfigure[]{\includegraphics[width=0.48\textwidth]{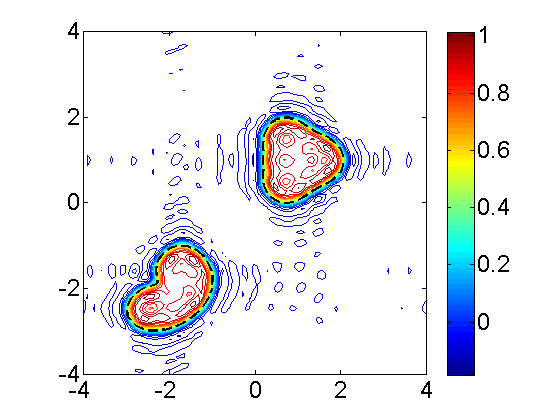}}
	\subfigure[]{\includegraphics[width=0.48\textwidth]{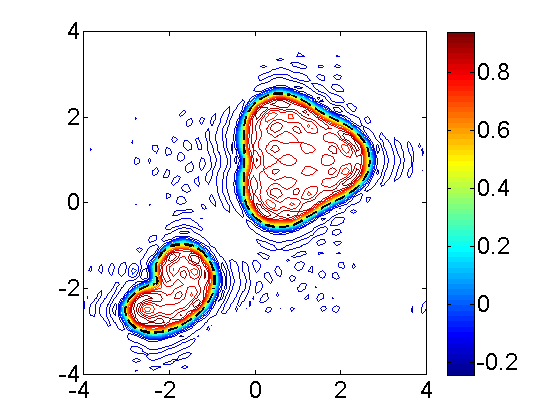}}
	\caption{\label{fig:mix} Contour plots of reconstructed mixed shape with different scale factors $(\beta_1,\beta_2)$:
		(a) $(1.13, 1.13)$, (b) $(1.94, 0.53)$, (c) $(1.88, 0.94)$, (d) $(1.96, 1.47)$. }
\end{figure}

\subsection{Rectangular Solid Experiments}
In the third example,  we verify the proposed method by using a set of artificial experiments on  rectangular solid.
The training dataset consists of $125$ rectangular solids with different
height, width and length. Here the height, width and length are uniformly distributed on $[1,2]$ with $5$ amounts, i.e., $1, 1.25, 1.5, 1.75, 2$.
Here, We consider four sets of different height, width and length of rectangular solids which are not covered by the training data.  The imaging results with different characteristic values are shown in Figure \ref{fig:rectangular}, where the black dotted lines denote the shadows of the exact cube boundary.  Due to discontinuities of the source, there
is Gibbs phenomena on the boundary of the rectangular solids.  On the whole, given the characteristic values, our proposed method is valid for determining the geometry shape.

\begin{figure}
	\subfigure[]{\includegraphics[width=0.45\textwidth]{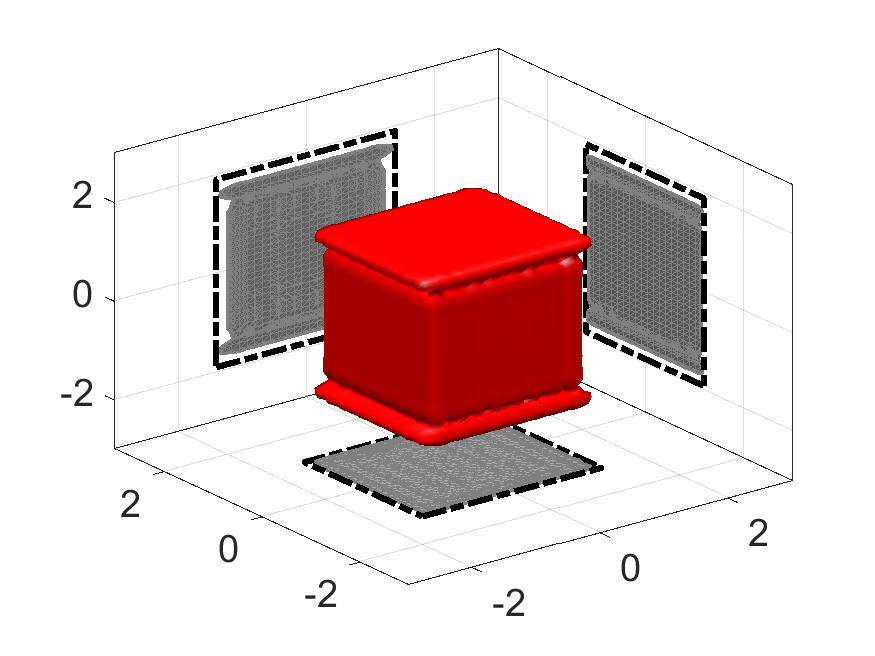}}
	\subfigure[]{\includegraphics[width=0.45\textwidth]{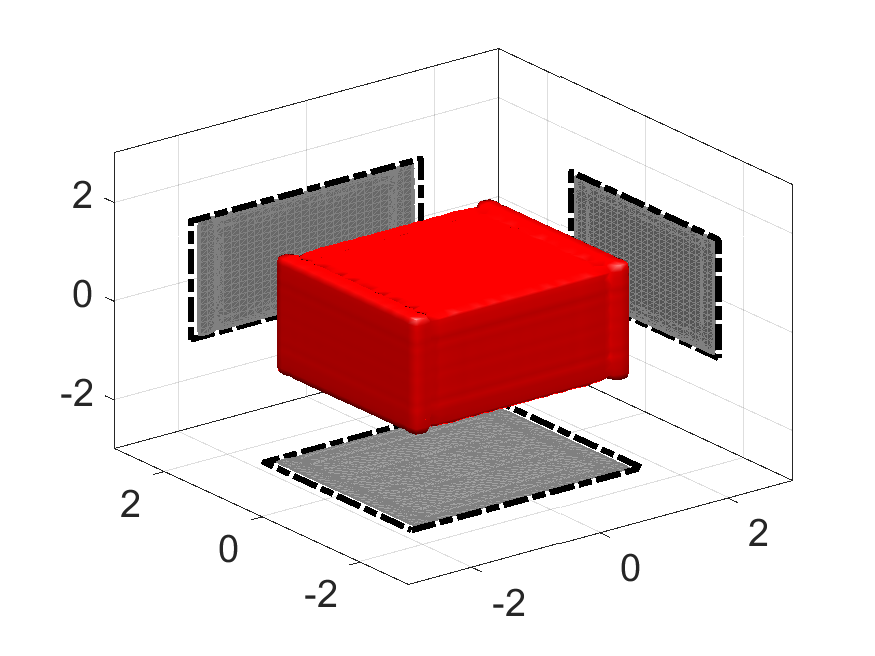}}\\
	\subfigure[]{\includegraphics[width=0.45\textwidth]{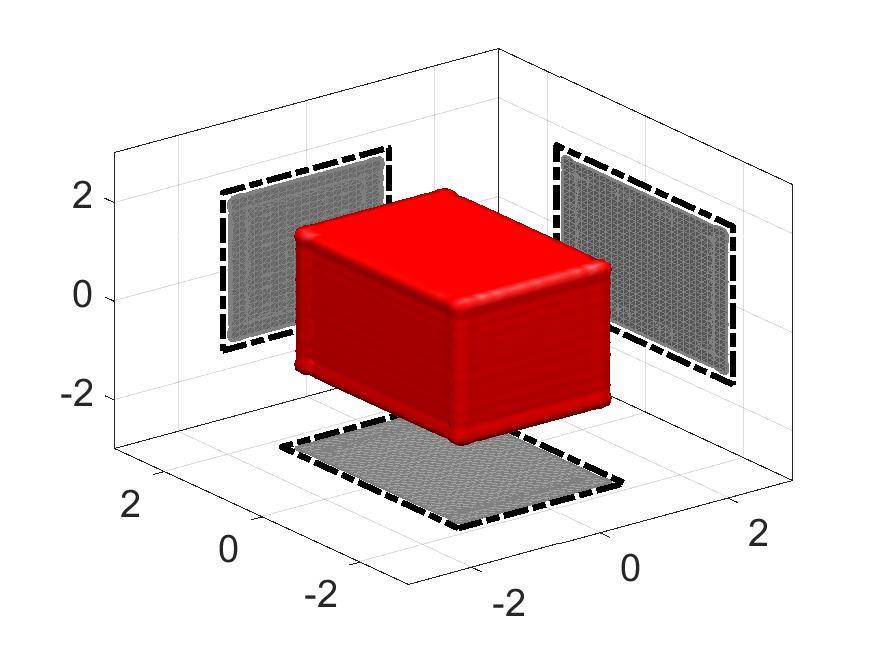}}
	\subfigure[]{\includegraphics[width=0.45\textwidth]{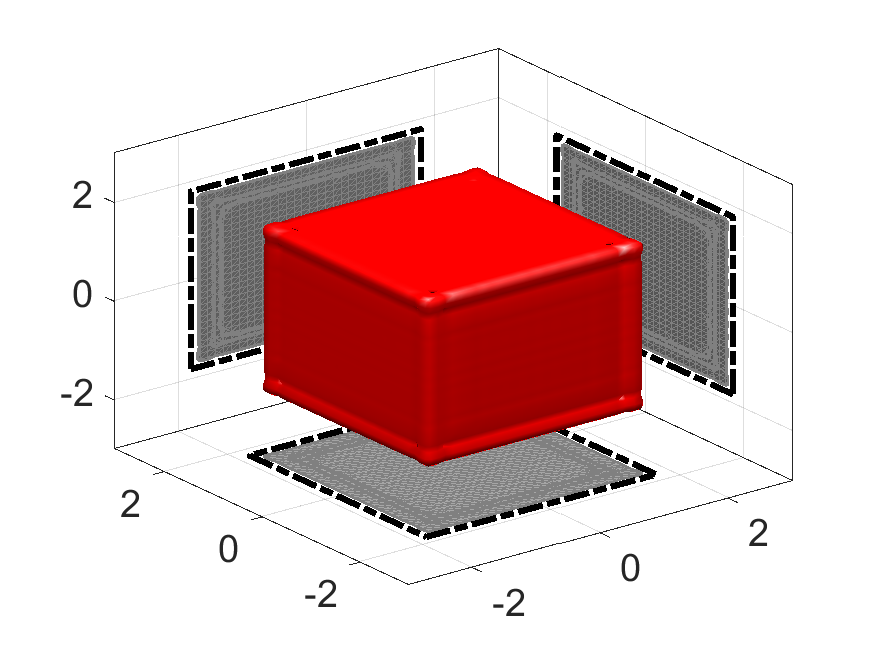}}
	\caption{\label{fig:rectangular} Isosurface plots of rectangular solid with different height, width and length, where the isosurface value is $1$. The sets of height, width and length are as follows: (a) $(1.9, 1.2, 1.4)$, (b) $( 1.2, 1.5,1.8)$, (c) $(1.6, 1.8, 1.3)$, (d) $(1.8, 1.8, 1.8)$. }
\end{figure}

\subsection{Body Experiments}In last example, we consider a challenging case and verify the proposed method by using a set of synthetic experiments on 3D human body shape.
The training dataset consists of $25$ bodies which are generated by the MakeHuman 1.1.1 soft. This experiments consider two characteristic values, i.e., height and relative weight. Define the exact weight by $EW$ and standard weight  by $SW$, then the
relative weight $RW$ is calculated by
\begin{equation*}
RW=\frac{EW}{SW}\times 100\%.
\end{equation*}
Here the height of the body is given by $1.5m, 1.6m, 1.7m, 1.8m, 1.9m$ and the relative weight of the body is given by $60\%, 80\%, 100\%, 120\%, 140\%$. Some human body shapes in the training dataset are presented in Figure \ref{fig:training set}.
In addition, we choose two characteristic values of human body which are not covered by the training data. The first body's height is $1.55m$ and the relative weight is $130\%$. The second body's height is $1.85m$ and the relative weight is $110\%$.
Figure \ref{fig:human1}(a) and Figure \ref{fig:human2}(a)  present the exact body shape with the given characteristic value.
Figure \ref{fig:human1}(b) and Figure \ref{fig:human2}(b) show the prediction of the human body shape with the given characteristic value. The results show that our method is efficient to predict the
human body shape.

\begin{figure}
	\subfigure[]{\includegraphics[width=0.32\textwidth]{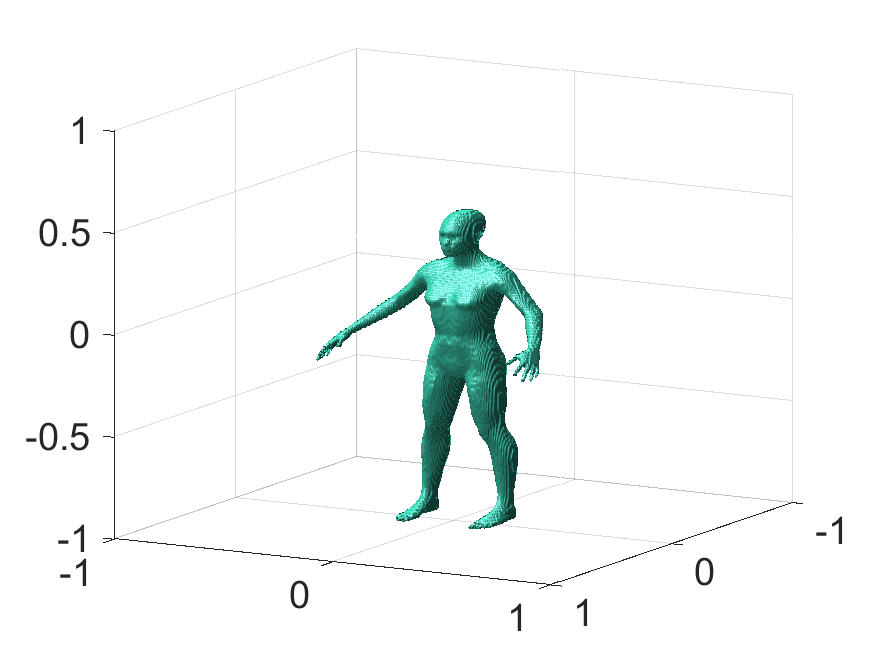}}
	\subfigure[]{\includegraphics[width=0.32\textwidth]{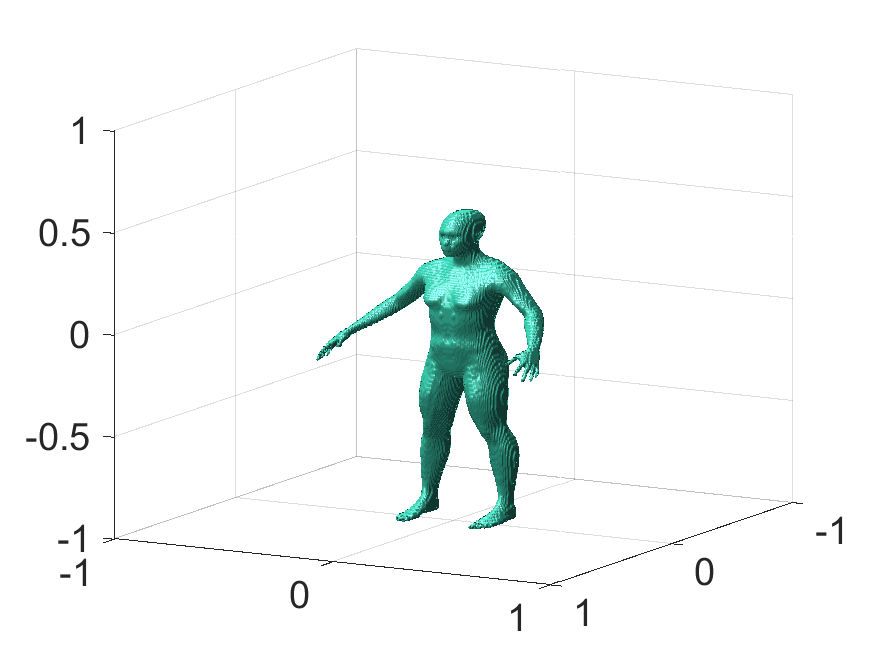}}
	\subfigure[]{\includegraphics[width=0.32\textwidth]{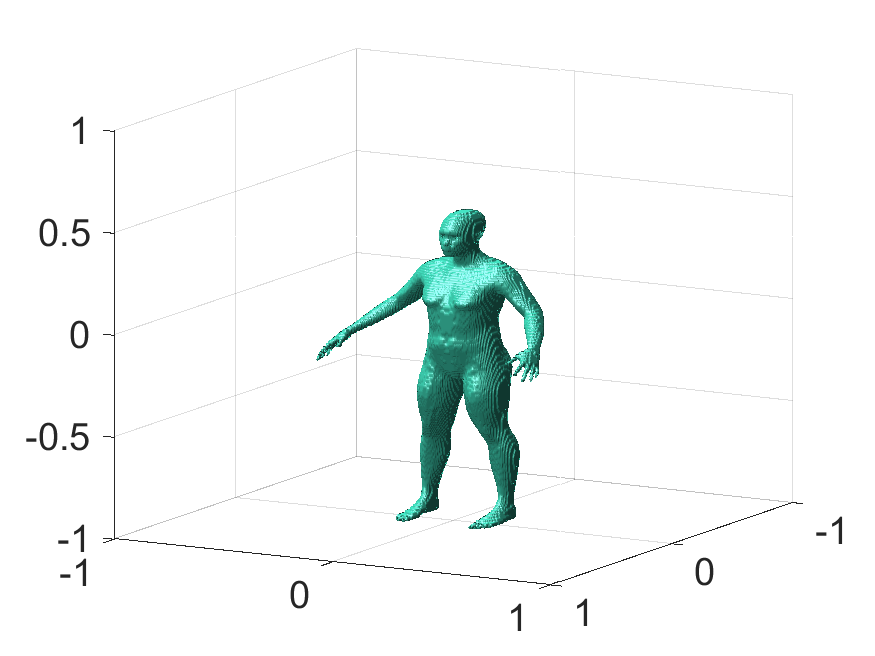}}\\
	\subfigure[]{\includegraphics[width=0.32\textwidth]{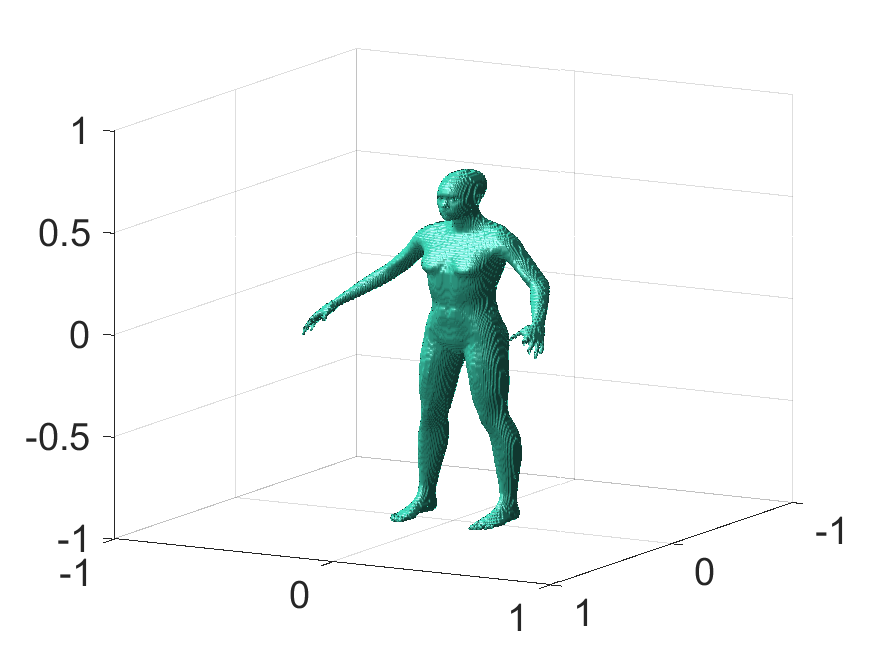}}
	\subfigure[]{\includegraphics[width=0.32\textwidth]{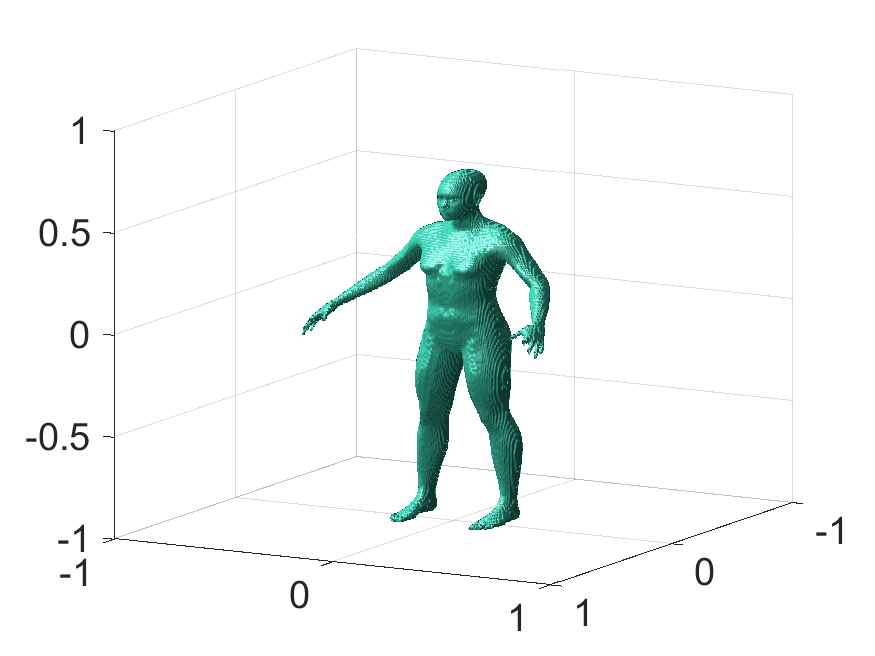}}
	\subfigure[]{\includegraphics[width=0.32\textwidth]{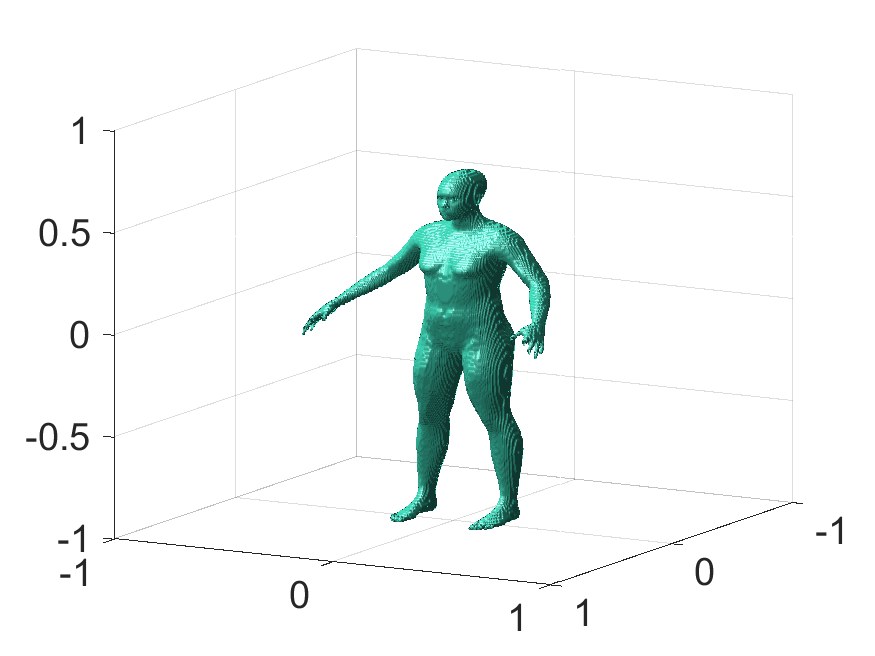}}\\
	\subfigure[]{\includegraphics[width=0.32\textwidth]{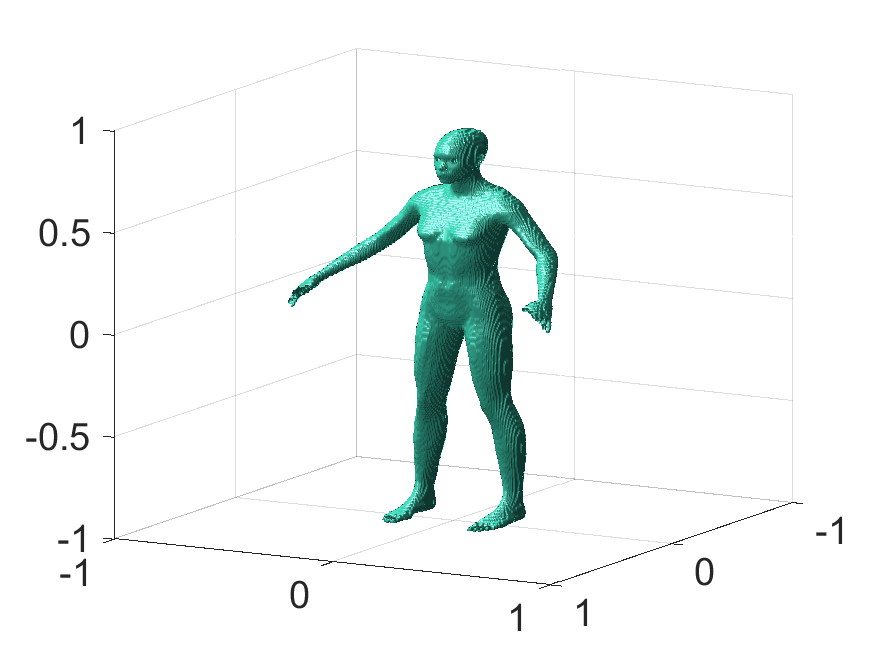}}
	\subfigure[]{\includegraphics[width=0.32\textwidth]{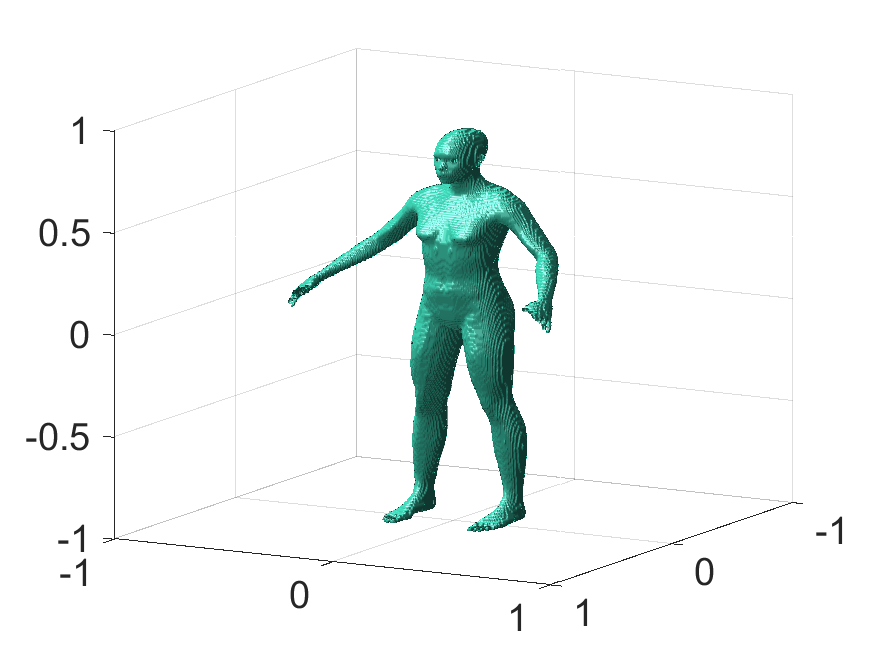}}
	\subfigure[]{\includegraphics[width=0.32\textwidth]{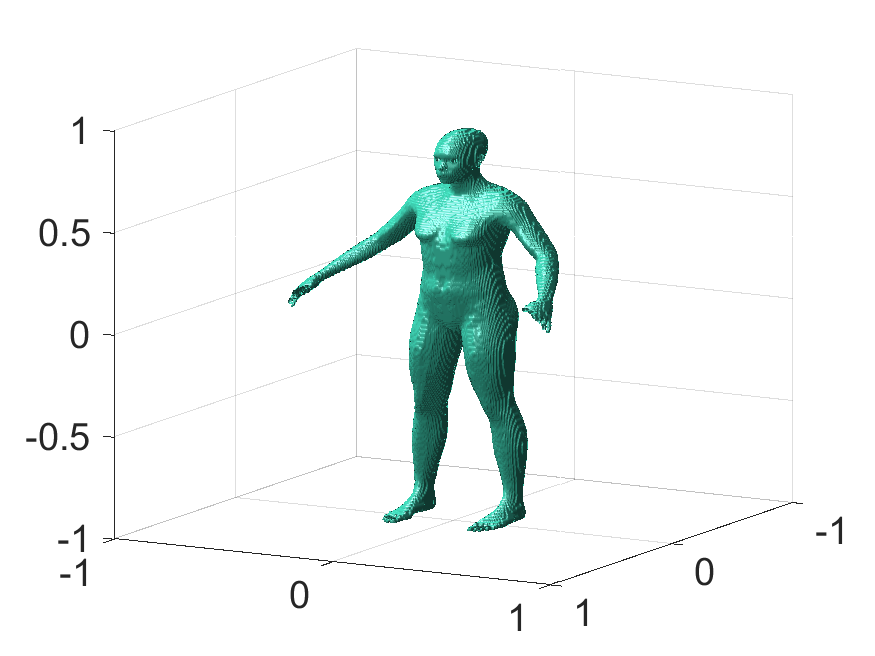}}\\
	\caption{\label{fig:training set} Isosurface plots of some training body data. Relative weight: the left column is $60\%$;
		the center column is $100\%$, the right column is $140\%$; height: the top row is $1.50m$, the center row is $1.70m$, the bottom row is $1.90m$. }
\end{figure}

\begin{figure}
	\subfigure[]{\includegraphics[width=0.49\textwidth]{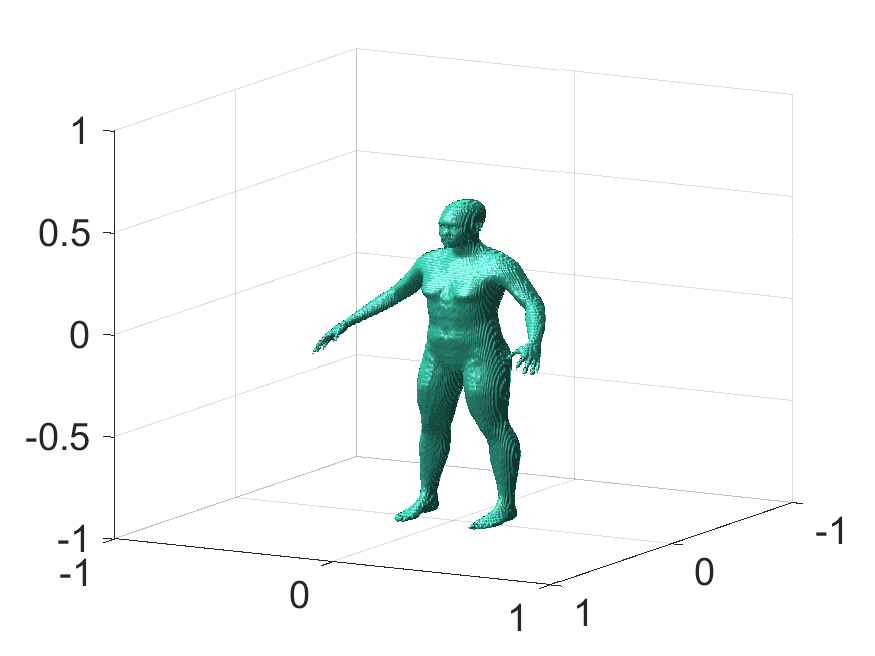}}
	\subfigure[]{\includegraphics[width=0.49\textwidth]{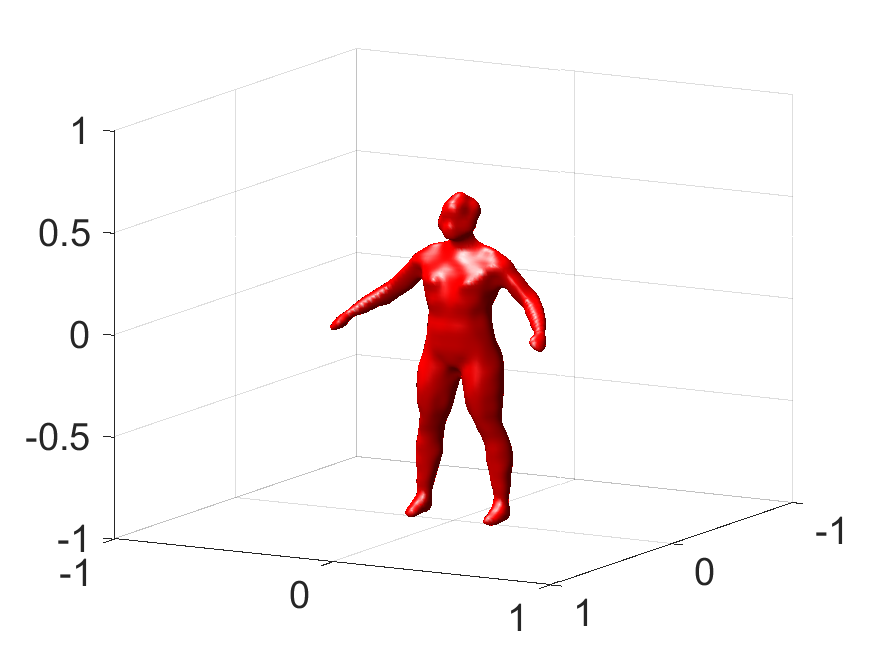}}
	\caption{\label{fig:human1}The relative weight is $130\%$, the height is $1.55 m$. (a) Exact body, (b) reconstruction body.}
\end{figure}

\begin{figure}
	\subfigure[]{\includegraphics[width=0.49\textwidth]{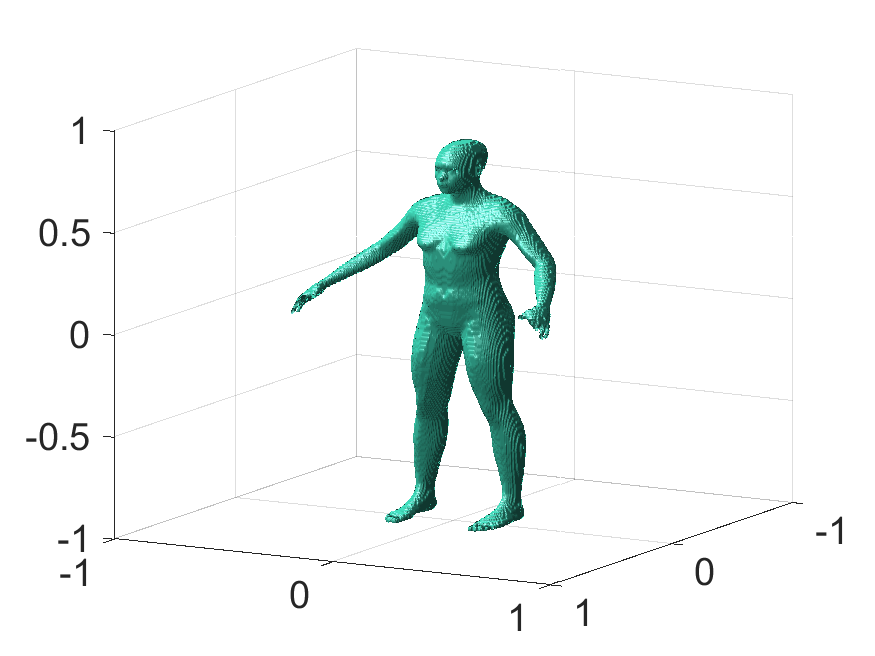}}
	\subfigure[]{\includegraphics[width=0.49\textwidth]{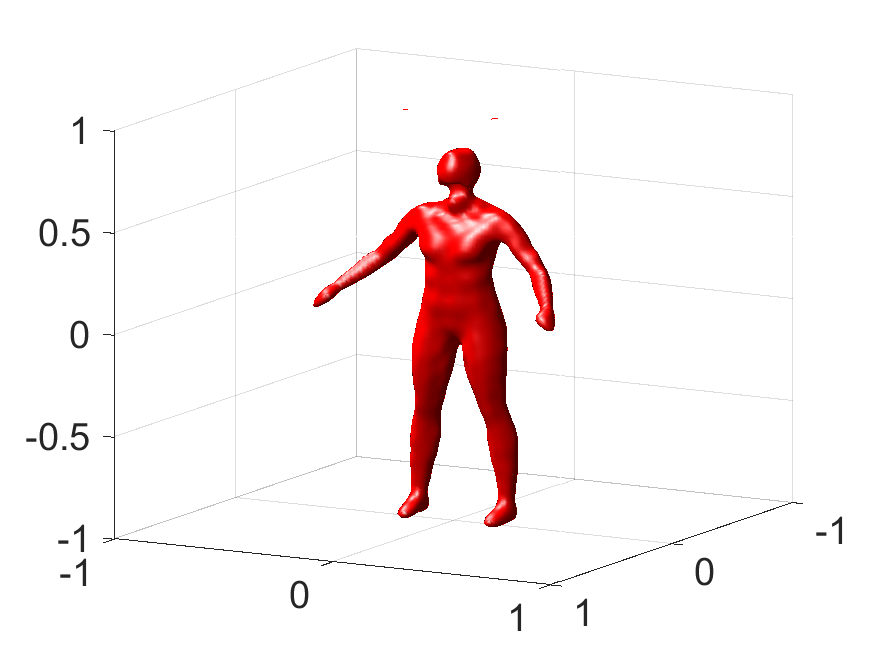}}
	\caption{\label{fig:human2} The relative weight is $110\%$, the height is $1.85 m$. (a) Exact body, (b) reconstruction body.}
\end{figure}

\section{Concluding remarks}

In this paper, we develop a machine-learning method in generating a geometric body shape through prescribing a set of characteristic values of the body. The generation is mainly based on a given training dataset consisting of certain pre-selected body shapes with statistically well-sampled characteristic values. A major novelty and critical ingredient of our study is the borrowing of inverse scattering techniques in the theory of wave propagation to the geometric shape generation. We introduce the notion of shape generator which establishes a one-to-one correspondence between the geometric shape space and the function space consisting of the multiple-frequency far-field patterns associated with the time-harmonic source scattering problem. The shape generator plays an intermediate role in the geometric shape generation. First, the training dataset of geometric shapes is converted into a subset of the function space consisting of the corresponding shape generators. Then a learning model is derived through a functional interpolation of the aforementioned shape generators. For a given set of characteristic values, one then uses the learning model to obtain the shape generator of the underlying geometric body and finally reconstructs it through a multiple-frequency Fourier method. 

To our best knowledge, the present study is the first attempt to introduce inverse scattering approaches in combination with machine learning to the geometric body generation and it opens up many opportunities for further developments. For example, in the current article, the shape generator is introduced through an inverse source scattering model where we make use of the one-to-one correspondence between a geometric shape and the multiple-frequency far-field pattern associated with a compactly-supported acoustic source. One may consider to introduce the shape generator through other inverse scattering models, e.g. the inverse acoustic obstacle scattering model (cf. \cite{LLZ}) or the inverse electromagnetic scattering model (cf. \cite{LLW}). In doing so, one may achieve other geometric shape generation schemes that are suitable for different applications.

\section*{Acknowledgment}
 The work of H. Liu was supported by the FRG and startup grants from Hong Kong Baptist University, Hong Kong RGC General Research Funds, 12302415 and 12302017.

\end{document}